\documentclass[final,1p,times]{elsarticle}

\newcounter{bla}

 \journal{Computer Physics Communications}

\usepackage{amsmath}
\usepackage{amsfonts}
\usepackage{amssymb}
\usepackage{stmaryrd}

\usepackage{feynmp}
\newcommand{\pd}{\partial}
\newcommand{\gog}{\mathfrak{g}}
\newcommand{\goh}{\mathfrak{h}}
\newcommand{\okappa}[1]{\mathcal{O}\left( \kappa^{#1} \right)}

\usepackage{amsthm}
\newtheorem{theorem}{Theorem}

\begin{document}

\begin{frontmatter}

\title{FeynGrav 4.0}

\author[a,b]{B. Latosh\corref{author}}

\cortext[author]{Corresponding author.\\\textit{E-mail address:} latosh@theor.jinr.ru}
\address[a]{Bogoliubov Laboratory of Theoretical Physics, JINR, Dubna 141980, Russia}
\address[b]{Dubna State University, Universitetskaya str. 19, Dubna 141982, Russia}

\begin{abstract}
  
  We present the new version of FeynGrav, a package that provides tools for working with Feynman rules for gravity models. The new version addresses two principal issues and improves the user experience. Firstly, we present a more sophisticated implementation of the BRST formalism for general relativity and quadratic gravity, producing a finite set of interaction rules between ghosts and gravitons. We also implement a higher-derivative gauge-fixing term for quadratic gravity. Secondly, we implement Feynman rules for Cheung-Remmen variables. These variables present the general relativity action in a polynomial form and produce a finite set of Feynman rules. Lastly, we introduce minor quality-of-life improvements to the package to enhance the user experience.


\noindent \textbf{NEW VERSION PROGRAM SUMMARY}

\begin{small}
  \noindent
      {\em Program Title:} FeynGrav\\
      {\em CPC Library link to program files:} (to be added by Technical Editor) \\
      {\em Developer's repository link:} https://github.com/BorisNLatosh/FeynGrav\\
      {\em Licensing provisions(please choose one):} GPLv3 \\
      {\em Programming language:} Wolfram Mathematica 10 and higher\\
      {\em Journal reference of previous version:} \cite{Latosh:2022ydd,Latosh:2023zsi,Latosh:2024lhl}\\
      {\em Does the new version supersede the previous version?:}*   \\
      The new version supersedes the previous version. The new version replaces a set of interaction rules between gravitons and the Faddeev-Popov ghost. The previous versions operated with an infinite set of rules, while the new version provides a finite set of interaction vertices. The new version also implements new interaction rules for special gravitational variables, resulting in a finite set of graviton vertices.
      {\em Reasons for the new version:*}\\
      The reason is the significant simplification of the calculation enabled by the newly implemented interaction rules.
      {\em Summary of revisions:}*\\
      Firstly, we revised the interaction between gravitons and Faddeev-Popov ghosts. We discovered a gauge fixing term that produces a finite set of the corresponding interaction vertices. We implemented this set of vertices in the new version. Secondly, we implemented a new higher derivative gauge fixing term that is widely used in gravity models with higher derivatives. Lastly, we implemented a set of interaction rules for the Cheung-Remmen variables. Feynman rules for such variables result in a finite set of interaction rules, which provides a huge simplification of calculations in the perturbation theory.
      {\em Nature of problem(approx. 50-250 words):}\\
      The perturbative approach to quantum gravity produces a theory with an infinite number of Feynman rules, which makes the theory extremely computationally challenging. In previous publications \cite{Latosh:2022ydd,Latosh:2023zsi,Latosh:2024lhl}, we found an algorithm which allows one to obtain such Feynman rules at any order of perturbation theory. However, it would be highly desirable to find a parametrisation of perturbative quantum gravity that would result in a theory with a finite set of Feynman rules, provided such a parametrisation exists.
      {\em Solution method(approx. 50-250 words):}\\
      We employ two new parametrisations that partially address the issues with an infinite number of Feynman rules. Firstly, we use geometric considerations to find a gauge fixing term within the BRST formalism that produces a finite set of Feynman rules for the theory. In that parameterisation, the ghost sector of the theory now has only a ghost propagator and a single ghost-graviton vertex. Therefore, we completely omitted the need to work with an infinite set of interaction rules for Faddeev-Popov ghosts. Secondly, we employed a non-linear parametrisation of general relativity. In that parameterisation, we also introduce an auxiliary field in such a way that the corresponding Lagrangian is polynomial in the field, which produces a finite set of Feynman rules. Therefore, the need to work with an infinite set of vertices is also eliminated for general relativity.

\end{small}
\end{abstract}
\end{frontmatter}

\section{Introduction}

Derivation of Feynman rules for gravity is a challenging task. Direct implementation of the perturbative approach to quantum gravity shows that the theory has an infinite number of Feynman rules \cite{DeWitt:1967yk,DeWitt:1967ub,DeWitt:1967uc,Iwasaki:1971vb,Veltman:1975vx,Donoghue:1994dn,Akhundov:1996jd,Prinz:2020nru,Latosh:2022ydd}. The previous publications \cite{Latosh:2022ydd,Latosh:2023zsi,Latosh:2024lhl} provided an algorithms that allow to calculate Feynman rules for all geometric theorises of gravity to any order. The algorithm was implemented in a package \texttt{FeynGrav} \cite{FeynGrav} based on the widely used \texttt{FeynCalc} package \cite{Shtabovenko:2023idz,Shtabovenko:2020gxv,Shtabovenko:2016sxi,Mertig:1990an} developed for \texttt{Wolfram Mathematica}.

This paper continues the development of such techniques. We focus on two particular aspects of development that are important from both the perspective of fundamental theory and computational efficiency.

First and foremost, we improve the Feynman rules for general relativity and quadratic gravity by employing the BRST formalism \cite{Becchi:1974md,Becchi:1974xu,Becchi:1975nq,Tyutin:1975qk,Faddeev:1967fc,Faddeev:1973zb}. In previous publications, we used the BRST formalism for general relativity \cite{Latosh:2023zsi}, but employed the Faddeev-Popov prescription for quadratic gravity \cite{Latosh:2024lhl}. On the practical level, both approaches are supposed to produce the same result, but the BRST formalism shows the underlying structure of the theory. We investigate the structure of the BRST complex and show that one can separate the gauge transformations associated with the true gauge redundancy of the gravitational field from transformations related to the freedom to choose a reference frame. We employed the standard perturbative expansion about a generic background spacetime. Separating a background geometry allowed us to distinguish the transformations associated with changes of coordinates in the background spacetime from the gauge transformations of the propagating metric perturbations. Such a separation naturally imposes constraints on the form of the gauge fixing term, leading to a significantly simpler structure of the interaction rules. With the specified gauge fixing term, the theory admits a single vertex describing the coupling between ghosts and gravitons. At the same time, vertices describing couplings between gravitons remain unaffected by the gauge fixing term. The structure of the Faddeev-Popov sector is in stark contrast with that obtained in previous publications, where ghosts admit an infinite number of interaction vertices.

We implemented the discussed BRST construction in general relativity and quadratic gravity. For quadratic gravity, we also address the higher derivative gauge fixing term. Such a gauge fixing term is widely used and found to be crucial for understanding the underlying structure of the theory \cite{Asorey:1996hz,Modesto:2017hzl}.

Secondly, we constructed Feynman rules for general relativity in the Cheung-Remmen variables \cite{Cheung:2017kzx}. These variables associate $\sqrt{-g} \, g^{\mu\nu}$ with the gravitational degrees of freedom. To be exact, variables $\sqrt{-g} \, g^{\mu\nu}$ were used long before the aforementioned publication. For instance, they were mentioned in the classical textbook \cite{Landau:1975pou}, the corresponding action was calculated back in \cite{Goldberg:1958zz}, and the first attempts to use such variables for quantisation trace back to \cite{Capper:1973pv} (see also \cite{buchbinder:1983a,McKeon:1994ds,Kalmykov:1994yj,Martellini:1997mu,Brandt:2015nxa,Brandt:2016eaj,McKeon:2020lqp,Brandt:2020vre,Brandt:2025lkd,Britto:2021pud}). We refer to these variables as ``Cheung-Remmen'' variables for two main reasons. Firstly, we used expressions from the corresponding paper \cite{Cheung:2017kzx}, which provides an extensive analysis of these variables. Secondly, many papers that considered such variables before, such as the aforementioned \cite{Capper:1973pv}, introduced auxiliary fields so that gravitons mix with them. The presence of such a term introduces additional complications to the perturbation theory, so we followed \cite{Cheung:2017kzx} where the mixing term is absent.

One introduces an auxiliary field, which brings the action into a form quadratic in the auxiliary field and linear in the derivatives of the Cheung-Remmen variables. After this, one performs a perturbative split of the Cheung-Remmen variables about the flat spacetime and introduces the corresponding BRST complex. The result is a set of Feynman rules of remarkable simplicity. It contains only three propagators, which are the propagator for metric perturbations, the propagator for the Faddeev-Popov ghost, and the propagator for the auxiliary field, which is pure algebraic. The interaction sector contains only four cubic vertices. They are a vertex for three metric perturbations; a vertex for two metric perturbations and one auxiliary field; a vertex for one metric perturbation and two auxiliary fields; and a vertex for one metric perturbation and two Faddeev-Popov ghosts.

Such a remarkable simplification comes with a price. Firstly, one can only perform the corresponding perturbative expansion about the flat spacetime. If one considers an arbitrary background, the theory receives an undesirable linear mixing between the auxiliary variable and the background metric. Secondly, there are no known generalisations of the Cheung-Remmen variables for higher derivative theories such as quadratic gravity or $f(R)$ gravity. Lastly, the propagating metric degrees of freedom admit a nonlinear algebraic relation with the standard metric perturbations. Consequently, the interaction between them and the matter degrees of freedom is also highly nonlinear and requires a separate study beyond the scope of this paper.

This paper is organised as follows. Section \ref{Section_BRST} discusses the BRST formalism for gravity. We introduce a background geometry to separate the true gauge transformations, associated with the gauge redundancy of the gravitational field, from the transformations associated with the freedom to choose a frame in the background geometry. We specify a gauge fixing term that produces a set of Feynman rules with a single cubic vertex describing the coupling between ghosts and gravitons. For quadratic gravity, we consider both the conventional and the higher derivative gauge fixing terms and present the corresponding set of Feynman rules for the ghost sector. Section \ref{Section_CR} discusses the Cheung-Remmen variables and the corresponding Feynman rules. We introduce these variables rigorously and prove a series of theorems relating them to the standard metric. We introduce the BRST complex and obtain the corresponding set of Feynman rules. Lastly, Section \ref{Section_FG} discusses the new changes to \texttt{FeynGrav}. Namely, we specify how the BRST formalism and the Cheung-Remmen variables are implemented. We also introduce minor quality-of-life improvements in the new version of the package. We conclude the paper in Section \ref{Section_Conclusion}, where we outline further development.

\section{BRST formalism for gravity}\label{Section_BRST}

This section recalls the construction of the (Becchi-Rouet-Stora-Tyutin) BRST formalism \cite{Becchi:1974md,Becchi:1974xu,Becchi:1975nq,Tyutin:1975qk} for gravity with a focus on general relativity and quadratic gravity. Since many publications (\cite{Nishijima:1978wq,Fradkin:1970pn,Faizal:2010dw,Ohta:2020bgz,Shapiro:2022ugk} and many more) discuss both the BRST formalism and the corresponding mathematical structure in great detail, we shall only discuss its most essential parts.

The BRST formalism is a more rigorous tool for introducing the Faddeev-Popov ghosts in the theory. It uses the BRST transformations, which are nilpotent transformations closely related to the gauge transformations. This allows one to introduce a gauge fixing term that breaks the gauge invariance of the action, but preserves its BRST invariance.

The main advantage of the formalism in the context of gravity is that it naturally suggests a gauge-fixing term, which significantly simplifies the set of Feynman rules. The previous publications \cite{Latosh:2022ydd,Latosh:2023zsi,Latosh:2024lhl} treated the gauge fixing with the Faddeev-Popov prescription and obtained a set of Feynman rules with an infinite number of ghost-graviton vertices. Below, we implement the BRST formalism for general relativity and introduce the simplest gauge fixing term naturally suggested by the formalism. In the corresponding set of Feynman rules, all the couplings between gravitons and the ghosts reduce to a single three-point vertex. Therefore, using a more sophisticated formalism results in a radically simpler set of interaction rules.

Let us proceed with the construction of the BRST formalism. Coordinate transformations play the role of gauge transformations in gravity. In turn, coordinate transformations are related to smooth vector fields. If one introduces a smooth vector field $\zeta$, then this field at each point shows the result of an infinitesimal coordinate transformation. To calculate the influence of such coordinate transformations, one shall use the Lie derivative $\mathcal{L}_\zeta$ with respect to the vector field $\zeta$. The axiomatic construction of the Lie derivative and its features are widely discussed in many textbooks on differential geometry \cite{spivak1979comprehensive,choquet1982analysis}.

The gauge transformation for the metric reads:
\begin{align}\label{the_metric_gauge_transfromation}
    \delta g_{\mu\nu} \overset{\text{def}}{=} \mathcal{L}_\zeta g_{\mu\nu} = \nabla_\mu \zeta_\nu + \nabla_\nu \zeta_\mu .
\end{align}
All the other geometric quantities are also subjected to gauge transformations spawned by this expression. The following theorem specifies such gauge transformations of some other geometric quantities.

\newpage

\begin{theorem}

    {~}

    \noindent Gauge transformations for gravitational theories are given by the following formulae.
    \begin{align}
        \begin{split}
            \delta g_{\mu\nu} &= \left( g_{\mu\lambda} \pd_\nu + g_{\nu\lambda} \pd_\mu + \pd_\lambda g_{\mu\nu} \right) \zeta^\lambda ,\\
            \delta \sqrt{-g} &= \sqrt{-g} \, \nabla^\mu \zeta_\mu ,\\
            \delta g^{\mu\nu} & = - \left( \nabla^\mu \zeta^\nu + \nabla^\nu \zeta^\mu  \right) ,\\
            \delta \Gamma^\alpha_{\mu\nu} &= \frac12\, \left( \nabla_\mu \nabla_\nu + \nabla_\nu \nabla_\mu \right) \zeta^\alpha + \frac12 \left( R_\mu{}^\alpha{}_{\nu\beta} + R_\nu{}^\alpha{}_{\mu\beta} \right) \zeta^\beta , \\
            \delta R_{\mu\nu\alpha\beta} & = \zeta^\sigma \nabla_\sigma R_{\mu\nu\alpha\beta} + R_{\sigma\nu\alpha\beta} \nabla_\mu \zeta^\sigma + R_{\mu\sigma\alpha\beta} \nabla_\nu \zeta^\sigma + R_{\mu\nu\sigma\beta} \nabla_\alpha \zeta^\sigma + R_{\mu\nu\alpha\sigma} \nabla_\beta \zeta^\sigma,\\
            \delta R_{\mu\nu} & = \zeta^\sigma \nabla_\sigma R_{\mu\nu} + R_{\mu\sigma} \nabla_\nu \zeta^\sigma + R_{\sigma\nu} \nabla_\mu \zeta^\sigma , \\
            \delta R &= \nabla_\mu R \, \zeta^\mu .
        \end{split}
    \end{align}
    
\end{theorem}

\noindent Let us note that since the Christoffel symbol is not a tensor, its variation is not a direct application of the  Lie derivative. Instead, the variation of the Christoffel symbol is its linear part with respect to the variation of the metric.

The Einstein-Hilbert action is invariant with respect to the change of coordinates, and its gauge transformation is a surface term:
\begin{align}
    \delta \int d^4 x \sqrt{-g} \left[ -\cfrac{2}{\kappa^2} \, R \right] = -\cfrac{2}{\kappa^2} \int d^4 x \, \pd_\mu \left[ \sqrt{-g} \,R\, \zeta^\mu \right].
\end{align}
Since the surface terms do not contribute to the field equations, the model enjoys a gauge invariance on the mass shell where the field equations are satisfied. It is also due to the gauge invariance that the quadratic part of the action is non-invertible, so one must introduce a gauge fixing term explicitly breaking the gauge invariance to construct a graviton propagator.

Let us axiomatically introduce the BRST transformation operator $\delta_B$, which shall satisfy two key requirements. Firstly, any quantity invariant with respect to the gauge transformations shall also be invariant with respect to the BRST transformations. Secondly, operator $\delta_B$ is nilpotent so that $\delta_B^2=0$.

One introduces a gauge fixing term using the BRST transformation as follows. One starts with a gauge invariant action $\mathcal{A}$, so the action is also BRST invariant:
\begin{align}
    \delta\mathcal{A} = 0 \Rightarrow \delta_B \mathcal{A}=0.
\end{align}
One introduces a gauge fixing term $\mathcal{A}_\text{gf}$ with is an exact BRST transformations of some other expression $\mathcal{A}_0$:
\begin{align}
    \mathcal{A}_\text{gf} & \overset{\text{def}}{=} \delta_B \mathcal{A}_0.
\end{align}
When the gauge fixing term is added to the action, it remains BRST invariant since the $\delta_B$ is nilpotent:
\begin{align}
    \delta_B \left( \mathcal{A} + \mathcal{A}_\text{gf} \right) = \delta_B \left( \mathcal{A} + \delta_B \mathcal{A}_0 \right) = \delta_B \mathcal{A} + \delta_B^2 \mathcal{A}_0 = 0 .
\end{align}
The relation between $\mathcal{A}_0$ and gauge invariance is more subtle. If one chooses $\mathcal{A}_0$ to be gauge invariant, it will also be BRST invariant, and the gauge fixing term vanishes. Consequently, the only non-trivial choice is to have such $\mathcal{A}_0$ that explicitly breaks the gauge invariance.

Let us follow the outlined approach. We introduce the following BRST operator $\delta_B$ to generalise the gauge transformations \eqref{the_metric_gauge_transfromation}:
\begin{align}\label{BRST_metric}
    \delta_B g_{\mu\nu} \overset{\text{def}}{=} \nabla_\mu c_\nu + \nabla_\nu c_\mu .
\end{align}
We introduce the ghost field $c_\mu$, which plays the role of a gauge parameter. We make the ghost field to have the canonical mass dimension $1$, so the BRST operator has the mass dimension $2$. This choice of dimensions is arbitrary and not uniquely fixed by any physical considerations. We demand that the ghost fields be anti-commuting with each other and with the BRST transformations. It is more useful to present \eqref{BRST_metric} in an equivalent form that does not rely on covariant derivatives:
\begin{align}
    \delta_B g_{\mu\nu} &= \left( g_{\mu\lambda} \pd_\nu + g_{\nu\lambda} \pd_\mu + \pd_\lambda g_{\mu\nu} \right) c^\lambda .
\end{align}
Since BRST transformations are nilpotent, expression \eqref{BRST_metric} uniquely fixes the BRST operator action on the ghost field:
\begin{align}\label{BRST_ghost}
    \begin{split}
        \delta_B c_\mu &= - c^\sigma \nabla_\mu c_\sigma , \\
        \delta_B c^\mu &= c^\sigma \nabla_\sigma c^\mu .
    \end{split}
\end{align}
Lastly, we introduce the antighost field $\overline{c}_\mu$ and the auxiliary field $B_\mu$, admitting the following BRST transformations:
\begin{align}\label{BRST_antighost}
    \begin{split}
        \delta_B \overline{c}_\mu &= B_\mu , \\
        \delta_B B_\mu &= 0.
    \end{split}
\end{align}
We also ensure the antighost has the canonical mass dimension $1$, so the auxiliary field $B_\mu$ has a non-canonical mass dimension $3$. This non-canonical mass dimension does not introduce any additional complication to the final result, because we will integrate the auxiliary field out, and it does not enter any final formulae.

\begin{theorem}

    {~}

    \noindent The following set of formulae provides the complete BRST transformation for gravity.
    \begin{align}\label{BRST_all}
        \begin{split}
            \delta_B g_{\mu\nu} & = \left( g_{\mu\lambda} \pd_\nu + g_{\nu\lambda} \pd_\nu + \pd_\lambda g_{\mu\nu} \right) c^\lambda , \\
            \delta_B c^\mu & = c^\sigma \nabla_\sigma c^\mu , \\
            \delta_B \overline{c}_\mu & = B_\mu  , \\
            \delta_B B_\mu &= 0.
        \end{split}
    \end{align}

\end{theorem}

The next step is to introduce a gauge fixing term. There are (at least) two ways to introduce a term that breaks the gauge invariance. The first way is to specify a dedicated reference frame in which the gauge fixing term takes a specific form. The second one is to specify a background geometry about which small metric perturbations propagate. The first way is more suitable for non-perturbative treatment, but it is much more sophisticated. The second approach is limited to perturbative treatment, as the derivation from the background geometry must be small. However, the technical side of this approach is much easier. This paper focuses on perturbative quantisation, but we discuss both gauge fixing procedures below.

Let us start with the gauge fixing term specified in a dedicated reference frame. One can introduce the gauge fixing term in the following form:
\begin{align}\label{Gauge_fixing_term_definition}
    \mathcal{A}_\text{gf} \overset{\text{def}}{=} \delta_B \int d^4 x \sqrt{-g} \, \overline{c}_\mu \, Y^{\mu\nu} \left( \mathcal{G}_\nu - \cfrac{\varepsilon}{2} \, \kappa^2 \, B_\nu \right).
\end{align}
In this expression, $Y^{\mu\nu}$ is a dimensionless invertible differential operator. It may contain only the metric and its derivatives, but not ghosts, antighosts, and auxiliary fields. We call $Y^{\mu\nu}$ a differential operator because it may contain derivatives, but this is not a necessary condition, and one can choose $Y^{\mu\nu}$ to be free from derivatives. The gauge fixing conditions $\mathcal{G}_\mu$ can depend only on the metric and its derivatives. We also ensure that it has the canonical mass dimension $1$, so the $\mathcal{A}_\text{gf}$ is dimensionless. 

Let us see how this term transforms under a coordinate transformation. We demand $\overline{c}_\mu$, $c_\mu$, and $B_\mu$ to be vectors, and for $Y^{\mu\nu}$ to be a tensor. On the contrary, the gauge fixing condition $\mathcal{G}_\mu$ must not be a vector, so the gauge fixing term can break the gauge invariance. For instance, one can use $ \mathcal{G}^\mu = g^{\alpha\beta} \Gamma^\mu_{\alpha\beta} $, which does not transform as a vector because of the Christoffel symbol transformation law. In that way, expression \eqref{Gauge_fixing_term_definition} is defined in one particular coordinate frame where all further computations shall take place.

We do not use this approach because it develops an undesirable coupling between the auxiliary field and ghosts. The BRST transformations \eqref{BRST_all} produce the following expression for the gauge fixing term:
\begin{align}\label{Gauge_fixing_term_1}
    \begin{split}
        \mathcal{A}_\text{gf} = \int d^4 x \sqrt{-g} \Bigg[ - \cfrac{\varepsilon}{2} \, \kappa^2\, B_\mu Y^{\mu\nu} B_\nu + B_\mu Y^{\mu\nu} \mathcal{G}_\nu - \overline{c}_\mu \left\{ \delta_B\left( Y^{\mu\nu} \mathcal{G}_\nu \right) + \nabla_\lambda c^\lambda Y^{\mu\nu} \mathcal{G}_\nu \right\} \\
        +\cfrac{\varepsilon}{2} \,\kappa^2 \, \overline{c}_\mu \left\{ \delta_B Y^{\mu\nu} + \nabla_\lambda c^\lambda \,Y^{\mu\nu} \right\} B_\nu \Bigg] .
    \end{split}
\end{align}
The last term introduces a coupling between the auxiliary field and the ghosts. Since the expression admits a part quadratic in the auxiliary field, one may expect to reduce the path integral to a Gaussian integral and integrate out the auxiliary field. In turn, the model experiences undesirable coupling between ghosts due to the last term. Because of this, we do not consider this way to introduce a gauge fixing term further.

This paper uses the gauge fixing procedure that specifies the background geometry and operates with small metric perturbations. The gauge fixing term shall use the integration over the background geometry $\int d^4 x \sqrt{-\overline{g}}$ defined by the background metric $\overline{g}_{\mu\nu}$, while $\overline{c}_\mu$, $c_\mu$, $B_\mu$, $\mathcal{G}_\mu$ must be vectors, and $Y^{\mu\nu}$ must be a tensor for the background geometry:
\begin{align}\label{Gauge_fixing_term_definition_perturbative}
    \mathcal{A}_\text{gf} \overset{\text{def}}{=} \delta_B \int d^4 x \sqrt{-\overline{g}} ~ \overline{c}_\mu \, Y^{\mu\nu} \left( \mathcal{G}_\nu - \cfrac{\varepsilon}{2} \, \kappa^2 \, B_\nu \right).
\end{align}
To exclude terms that can introduce coupling between the ghost and auxiliary field, we require $Y^{\mu\nu}$ to be BRST invariant
\begin{align}
    \delta_B Y^{\mu\nu} = 0, 
\end{align}
so it can only depend on the background metric. On the contrary, the gauge fixing conditions $\mathcal{G}_\mu$ can depend on background metrics and perturbations.

We proceed with this choice of the gauge fixing term and split the complete metric into the background metric $g_{\mu\nu}$ and small metric perturbations $h_{\mu\nu}$:
\begin{align}\label{the_perturbative_metric}
    g_{\mu\nu} = \overline{g}_{\mu\nu} + \kappa \, h_{\mu\nu} .
\end{align}
Perturbations $h_{\mu\nu}$ transform as a tensor with respect to the background geometry. The background metric has no dynamical degrees of freedom, so it is BRST invariant.
\begin{align}
    \delta_B \overline{g}_{\mu\nu} = 0.
\end{align}

Let us especially highlight that BRST transformations \eqref{BRST_all} are not related to the change of coordinates on the background spacetime. The original gauge transformations \eqref{the_metric_gauge_transfromation} and the corresponding BRST transformations \eqref{BRST_all} are due to the redundant description of gravitational degrees of freedom. The coordinate transformations on the background spacetime are due to the geometrical nature of the formalism and cannot be attributed to any dynamic degree of freedom. To put it otherwise, the introduction of the background geometry explicitly separates the redundancy associated with the freedom to choose a coordinate frame from the true gauge redundancy of the gravitational field associated with the presence of dynamical constraints.

To proceed with the perturbative treatment, we shall rewrite the original BRST transformations \eqref{BRST_all} in terms of the background geometry.

\begin{theorem}

    {~}

    \noindent BRST transformations \eqref{BRST_all} take the following form for the perturbative expansion \eqref{the_perturbative_metric}.
    \begin{align}\label{BRST_perturbative}
        \begin{split}
            & \delta_B h_{\mu\nu} = \cfrac{1}{\kappa} \left( \overline{\nabla}_\mu c_\nu + \overline{\nabla}_\nu c_\mu \right) + \left( h_{\mu\lambda} \overline{\nabla}_\nu + h_{\nu\lambda}\overline{\nabla}_\mu + \overline{\nabla}_\lambda h_{\mu\nu} \right) c^\lambda ,\\
            & \delta_B c^\mu = c^\sigma \overline{\nabla}_\sigma c^\mu ,\\
            & \delta_B \overline{c}_\mu = B_\mu ,\\
            & \delta_B B_\mu = 0.
        \end{split}
    \end{align}
    Here, $\overline{\nabla}$ is the covariant derivative for the background geometry.

\end{theorem}

\noindent These transformations allow one to obtain an explicit expression for the gauge fixing term.

\begin{theorem}

    {~}

    \noindent The gauge fixing term \eqref{Gauge_fixing_term_definition_perturbative} takes the following form if the BRST transformation \eqref{BRST_perturbative} holds:
    \begin{align}
        \mathcal{A}_\text{gf} = \int d^4 x \sqrt{-\overline{g}} \Bigg[ - \cfrac{\varepsilon}{2} \, \kappa^2 B_\mu Y^{\mu\nu} B_\nu + B_\mu Y^{\mu\nu} \mathcal{G}_\nu - \overline{c}_\mu Y^{\mu\nu} \delta_B \mathcal{G}_\nu \Bigg] .
    \end{align}
    
\end{theorem}

\noindent The corresponding generating functional $\mathcal{Z}$, defined via the functional integral, also simplifies. The simplification occurs because the functional integral is a Gaussian integral with respect to the auxiliary field $B_\mu$. 

\begin{theorem}

    {~}

    \noindent The generating functional of the following form
    \begin{align}
        \mathcal{Z} & \overset{\text{def}}{=} \int\mathcal{D}[g]\mathcal{D}[B]\mathcal{D}[c]\mathcal{D}[\overline{c}] \exp\Bigg[ i\,\mathcal{A}[g] + i \, \mathcal{A}_\text{gf}[g,B,c,\overline{c}] \Bigg]
    \end{align}
    with the gauge fixing term \eqref{Gauge_fixing_term_definition_perturbative} takes the following form after the integration over the auxiliary field
    \begin{align}\label{Generating_Functional_Generic_1}
        \mathcal{Z} & = \cfrac{1}{\sqrt{\det Y}} \, \int \mathcal{D}[h]\mathcal{D}[c]\mathcal{D}[\overline{c}] \exp\Bigg[ i \, \mathcal{A}[h] + i \int d^4 x \sqrt{-\overline{g}} \left[ \cfrac{1}{2\,\varepsilon\,\kappa^2} ~ \mathcal{G}_\mu Y^{\mu\nu} \mathcal{G}_\nu -  \overline{c}_\mu Y^{\mu\nu} \delta_B \mathcal{G}_\nu \right] \Bigg] .
    \end{align}

\end{theorem}

The presence of $1/\sqrt{\det Y}$ is well known, and it was discovered in many classical papers on higher derivative gravity  \cite{Julve:1978xn,Fradkin:1981hx,Barth:1983hb,Ohta:2020bgz,Shapiro:2022ugk}. However, the factor $1/\sqrt{\det Y}$ is non-trivial if and only if $Y^{\mu\nu}$ is a differential operator. If $Y^{\mu\nu}$ does not contain derivatives, then the corresponding determinant is a finite constant absorbed by the normalisation. Let us also note that the sign of the last term is entirely arbitrary. One can change it either by changing the sign of the gauge fixing condition $\mathcal{G}_\mu \to - \mathcal G_\mu$ or by simultaneously changing the sign of the $Y^{\mu\nu} \to - Y^{\mu\nu}$ and the gauge fixing parameter $\varepsilon \to - \varepsilon$.

We use the following gauge fixing condition, which is a vector for the background geometry:
\begin{align}\label{the_gauge_fixing_condition}
    & \mathcal{G}_\mu \overset{\text{def}}{=} \kappa\left( \overline{\nabla}^\sigma h_{\mu\sigma} - \frac12 \, \overline{\nabla}_\mu \left( \overline{g}^{\alpha\beta} h_{\alpha\beta} \right) \right) = \cfrac{\kappa}{2} \, \mathfrak{C}_\mu{}^{\nu\alpha\beta} \overline{\nabla}_\nu h_{\alpha\beta}
\end{align}
with the following notation of the $\mathfrak{C}$ tensor:
\begin{align}
    \mathfrak{C}_{\mu\nu\alpha\beta} \overset{\text{def}}{=} \overline{g}_{\mu\alpha} \overline{g}_{\nu\beta} + \overline{g}_{\mu\beta} \overline{g}_{\nu\alpha} - \overline{g}_{\mu\nu} \overline{g}_{\alpha\beta}.
\end{align}
The BRST action on this gauge fixing term has the following form:
\begin{align}
    \delta_B \mathcal{G}_\mu = \overline{\square} c_\mu + \overline{R}_{\mu\nu} \, c^\nu + \cfrac{\kappa}{2}\, \overline{\nabla}_\nu \left[ \mathfrak{C}_\mu{}^{\nu\alpha\beta} \left( h_{\alpha\lambda} \overline{\nabla}_\beta + h_{\beta\lambda}\overline{\nabla}_\alpha + \overline{\nabla}_\lambda h_{\alpha\beta} \right) c^\lambda \right] .
\end{align}
Consequently, we obtain the following form of the generating functional that we use below.

\begin{theorem}

    {~}

    \noindent The generating functional\footnote{ We took  formulae for \eqref{Generating_Functional_Generic_1} and replaced $\mathcal{G}_\mu \to - \mathcal{G}_\mu$ and $\varepsilon\to\varepsilon/4$. If we were to proceed with \eqref{Generating_Functional_Generic_1}, then the gauge fixing term for $h_{\mu\nu}$ would receive an additional $1/4$ factor, while the ghost-perturbation vertex would have a different sign. The choice of gauge fixing term and parameter does not influence the physical content of the model. We only change these definitions to obtain more conveniently looking formulae.}
    \begin{align}
        \begin{split}
            \mathcal{Z} & \overset{\text{def}}{=} \int\mathcal{D}[g]\mathcal{D}[B]\mathcal{D}[c]\mathcal{D}[\overline{c}] \exp\Bigg[ i\,\mathcal{A}[g] + i \, \mathcal{A}_\text{gf}[g,B,c,\overline{c}] \Bigg]\, \\
            \mathcal{A}_\text{gf} &\overset{\text{def}}{=} \delta_B \int d^4 x \sqrt{-\overline{g}} ~ \overline{c}_\mu \, Y^{\mu\nu} \left( - \mathcal{G}_\nu - 2\, \varepsilon \, \kappa^2 B_\nu \right) \\
            \mathcal{G}_\mu &\overset{\text{def}}{=} \kappa\left( \overline{\nabla}^\sigma h_{\mu\sigma} - \frac12 \, \overline{\nabla}_\mu \left( \overline{g}^{\alpha\beta} h_{\alpha\beta} \right) \right)
        \end{split}
    \end{align}
    takes the following form after the integration over the auxiliary field
    \begin{align}\label{the_Generating_functional}
        \begin{split}
            \mathcal{Z} =& \cfrac{1}{\sqrt{\det Y}} \int\mathcal{D}[h]\mathcal{D}[c]\mathcal{D}[\overline{c}] \exp\Bigg[ i\, \mathcal{A} + i \int d^4 x \sqrt{-\overline{g}} \left[ \cfrac{1}{2\,\varepsilon} \overline{\nabla}_\rho h_{\mu\nu} \left( \mathfrak{C}_\lambda{}^{\rho\mu\nu} \mathfrak{C}_\tau{}^{\sigma\alpha\beta} \,Y^{\lambda\tau} \right) \overline{\nabla}_\sigma h_{\alpha\beta} \right] \\
            & + i \! \int \!\! d^4 x \sqrt{-\overline{g} } \left[ \overline{c}_\mu \,Y^{\mu\nu} \overline{\square} c_\nu  + \overline{c}_\mu \,Y^{\mu\nu} \overline{R}_{\nu\lambda} c^\lambda + \cfrac{\kappa}{2} \,\overline{c}_\mu \, Y^{\mu\nu} \overline{\nabla}_\tau \left[ \mathfrak{C}_\nu{}^{\tau\alpha\beta} \left( h_{\alpha\lambda} \overline{\nabla}_\beta \!+\! h_{\beta\lambda}\overline{\nabla}_\alpha \!+\! \overline{\nabla}_\lambda h_{\alpha\beta} \right) c^\lambda \right] \right] \Bigg].
        \end{split}
    \end{align}
\end{theorem}

The generating functional \eqref{the_Generating_functional} has a remarkably simple structure. Its first term is quadratic in perturbations, so it only influences the form of the graviton propagator. The expression \eqref{the_Generating_functional} has no terms with higher powers of perturbations, so this gauge fixing procedure does not influence expressions for graviton vertices. The ghost sector contains a part quadratic in ghosts, describing their propagation around the background spacetime. Only the last term in \eqref{the_Generating_functional} describes an interaction which is a cubic coupling between two ghosts and a single graviton. This simple structure of the ghost sector is in stark contrast with the results obtained with the direct implementation of the Faddeev-Popov prescription \cite{Latosh:2023zsi}. The reason behind this is the suitable choice of the gauge fixing term, which appears to be the simplest gauge fixing term conceivable within the BRST formalism.

The generating functional \eqref{the_Generating_functional} provides the most general expression suitable for any gravitational model. To operate with a particular model, one has to specify the microscopic action $\mathcal{A}$ and the operator $Y^{\mu\nu}$. We consider the two most important cases: general relativity and quadratic gravity. These models play a unique role since they define the most general structure of a graviton propagator. Any other gravity model reduces to them up to a non-trivial coupling between matter and gravity.

\subsection{General relativity}

The structure of the Hilbert action $\mathcal{A}_\text{H}$ was thoroughly discussed in previous publications \cite{Latosh:2022ydd,Latosh:2023zsi,Latosh:2024lhl}, so we do not reproduce its derivation here. The main focus of this subsection is the gauge fixing term and the ghost sector.

General relativity admits second-order field equations, so there is no need to introduce derivatives to $Y^{\mu\nu}$. We define $Y^{\mu\nu}$ simply to be the background metric:
\begin{align}
    Y^{\mu\nu} = \overline{g}^{\mu\nu}.
\end{align}
This choice of the operator produces the following structure of the generating functional \eqref{the_Generating_functional}:
\begin{align}
    \begin{split}
        \mathcal{Z} =& \int\mathfrak{D}[h]\mathcal{D}[c]\mathcal{D}[\overline{c}]\,\exp\Bigg[ i\, \left(-\cfrac{2}{\kappa^2} \right) \int d^4 x \sqrt{-g} \,g^{\mu\nu} g^{\alpha\beta} g^{\rho\sigma} \left[ \Gamma_{\alpha\mu\rho} \Gamma_{\sigma\nu\beta} - \Gamma_{\alpha\mu\nu} \Gamma_{\rho\beta\sigma} \right] \\
        &- \cfrac{i}{2\,\varepsilon} \int d^4 x \sqrt{-\overline{g}} \,h^{\mu\nu} \Big[ \left( \overline{g}_{\nu\rho} \overline{g}_{\sigma\alpha} \overline{g}_{\beta\mu} +\cdots \right) - \left( \overline{g}_{\mu\nu} \overline{g}_{\alpha\rho} \overline{g}_{\beta\sigma} + \cdots  \right) + \overline{g}_{\mu\nu} \overline{g}_{\alpha\beta} \overline{g}_{\rho\sigma} \Big] \overline{\nabla}_\rho \overline{\nabla}_\sigma h^{\alpha\beta} \\
        & + i\, \int d^4 x \sqrt{-\overline{g}} \Big[ \overline{c}^\nu \left( \overline{g}_{\mu\nu} \overline{\square} + \overline{R}_{\mu\nu} \right) c^\mu - \cfrac{\kappa}{2} \,\overline{\nabla}_\lambda \overline{c}^\nu \mathfrak{C}_\nu{}^{\lambda\alpha\beta} \left( h_{\alpha\mu} \overline{\nabla}_\beta + h_{\beta\mu} \overline{\nabla}_\alpha + \overline{\nabla}_\mu h_{\alpha\beta} \right) c^\mu \Big] \Bigg].
    \end{split}
\end{align}
In the second term, $\cdots$ note the terms required for the expression to be symmetric for $\mu\leftrightarrow\nu$, $\alpha\leftrightarrow\beta$, and $(\mu,\nu)\leftrightarrow(\alpha,\beta)$. In turn, $\Gamma_{\alpha\mu\nu}$ is the Christoffel symbol with all low indices, which admits a finite expansion in perturbations
\begin{align}
    \Gamma_{\alpha\mu\nu} = \cfrac12\,\left[ \pd_\mu g_{\nu\alpha} + \pd_\nu g_{\mu\alpha} - \pd_\alpha g_{\mu\nu} \right] = \overline{\Gamma}_{\alpha\mu\nu} + \kappa\, \overline{\Gamma}^\sigma_{\mu\nu} \,h_{\sigma\alpha} + \cfrac{\kappa}{2} \left[ \overline{\nabla}_\mu h_{\nu\alpha} + \overline{\nabla}_\nu h_{\mu\alpha} - \overline{\nabla}_\alpha h_{\mu\nu} \right] .
\end{align}

To obtain the Feynman rules, we choose the flat background and obtain the following expression:
\begin{align}\label{Generating_functional_GR}
    \begin{split}
        \mathcal{Z} =& \int\mathfrak{D}[h]\mathcal{D}[c]\mathcal{D}[\overline{c}]\,\exp\Bigg[ i\, \left(-\cfrac{2}{\kappa^2} \right) \int d^4 x \sqrt{-g} \,g^{\mu\nu} g^{\alpha\beta} g^{\rho\sigma} \left[ \Gamma_{\alpha\mu\rho} \Gamma_{\sigma\nu\beta} - \Gamma_{\alpha\mu\nu} \Gamma_{\rho\beta\sigma} \right] \\
        & + i\, \cfrac{1}{\varepsilon} \int \cfrac{d^4 p}{(2\pi)^4} \,h^{\mu\nu}(-p) \left( P^1_{\mu\nu\alpha\beta} +\cfrac32\,P^0_{\mu\nu\alpha\beta} + \cfrac12\, \overline{P}^0_{\mu\nu\alpha\beta} - \cfrac12 \overline{\overline{P}}^0_{\mu\nu\alpha\beta} \right) \, p^2 \,h^{\alpha\beta}(p) \\
        & + i\, \int\cfrac{d^4p}{(2\pi)^4} \, \overline{c}^\nu(-p) \left( - \eta_{\mu\nu} p^2 \right) c^\mu(p) - i\, \cfrac{\kappa}{2}\, \int d^4 x \, \pd_\lambda \overline{c}^\nu \mathfrak{C}_\nu{}^{\lambda\alpha\beta} \left( h_{\alpha\mu} \pd_\beta + h_{\beta\mu} \pd_\alpha + \pd_\mu h_{\alpha\beta} \right) c^\mu \Big] \Bigg].
    \end{split}
\end{align}
Here $P^1$, $P^2$, $P^0$, $\overline{P}^0$, and $\overline{\overline{P}}^0$ are the Nieuwenhuizen operators \cite{VanNieuwenhuizen:1973fi,Accioly:2000nm} which we present in Appendix \ref{Appendix_Nieuwenhuizen}. This expression produces the standard ghost propagator:
\begin{align}
    \begin{gathered}
        \begin{fmffile}{Ghost_Propagator}
            \begin{fmfgraph*}(30,30)
                \fmfleft{L}
                \fmfright{R}
                \fmf{dots_arrow}{L,R}
                \fmflabel{$\mu$}{L}
                \fmflabel{$\nu$}{R}
            \end{fmfgraph*}
        \end{fmffile}
    \end{gathered}
     \hspace{20pt }= i\, \cfrac{\eta_{\mu\nu}}{p^2} \, .
\end{align}
The ghost couples to gravitons by a single three-particle vertex:
\begin{align}
\nonumber \\
    \begin{gathered}
        \begin{fmffile}{Ghost_Graviton_Vertex}
            \begin{fmfgraph*}(40,40)
                \fmfleft{L}
                \fmfright{R1,R2}
                \fmf{dbl_wiggly}{L,V}
                \fmf{dots_arrow}{R1,V}
                \fmf{dots_arrow}{V,R2}
                \fmfdot{V}
                \fmflabel{$\rho\sigma,k$}{L}
                \fmflabel{$\mu,p_1$}{R1}
                \fmflabel{$\nu,p_2$}{R2}
            \end{fmfgraph*}
        \end{fmffile}
    \end{gathered}    
    = i\,\cfrac{\kappa}{2} \, (p_2)_\lambda \mathfrak{C}_\nu{}^{\lambda\alpha\beta}\left[ (p_1)_\alpha I_{\beta\mu\rho\sigma} + (p_1)_\beta I_{\alpha\mu\rho\sigma} + k_\mu I_{\alpha\beta\rho\sigma} \right] . \\ \nonumber
\end{align}
In this diagram, all momenta are directed inwards and $I_{\mu\nu\alpha\beta}$ is the unit tensor:
\begin{align}
    I_{\mu\nu\alpha\beta} \overset{\text{def}}{=} \cfrac12\,\left( \eta_{\mu\alpha}\eta_{\nu\beta} + \eta_{\mu\beta} \eta_{\nu\alpha} \right).
\end{align}
Lastly, the quadratic part of the action, which is responsible for the graviton propagator, reads:
\begin{align}
    \begin{split}
        & -\cfrac{2}{\kappa^2} \int d^4 x \sqrt{-g} \,g^{\mu\nu} g^{\alpha\beta} g^{\rho\sigma} \left[ \Gamma_{\alpha\mu\rho} \Gamma_{\sigma\nu\beta} - \Gamma_{\alpha\mu\nu} \Gamma_{\rho\beta\sigma} \right] \\
        & \hspace{20pt}+ \cfrac{1}{\varepsilon} \int \cfrac{d^4 p}{(2\pi)^4} \,h^{\mu\nu}(-p) \left( P^1_{\mu\nu\alpha\beta} +\cfrac32\,P^0_{\mu\nu\alpha\beta} + \cfrac12\, \overline{P}^0_{\mu\nu\alpha\beta} - \cfrac12  \,\overline{\overline{P}}^0_{\mu\nu\alpha\beta} \right) \, p^2 \,h^{\alpha\beta}(p) \\
        & = \cfrac12 \, \int \cfrac{d^4 p}{(2\pi)^4}  \, h^{\mu\nu}(-p) \left( \mathcal{O}_\text{GR+gf} \right)_{\mu\nu\alpha\beta} \,p^2\, h^{\alpha\beta}(p) + \mathcal{O}\left(\kappa \right) .
    \end{split}
\end{align}
We omit $\mathcal{O}\left(\kappa\right)$ terms because they describe graviton vertices and remain unaffected by the gauge fixing term. The operator $\mathcal{O}_\text{GR+gf}$ has the following form\footnote{To obtain expressions from \cite{Latosh:2023zsi}, one shall replace $\varepsilon \to 4/\epsilon$.} in terms of the Nieuwenhuizen operators:
\begin{align}
    \left( \mathcal{O}_\text{GR+gf} \right)_{\mu\nu\alpha\beta}  = \cfrac{2}{\varepsilon} \, P^1_{\mu\nu\alpha\beta} + P^2_{\mu\nu\alpha\beta} + \left( \cfrac{3}{\varepsilon} - 2 \right) P^0_{\mu\nu\alpha\beta} + \cfrac{1}{\varepsilon} \, \overline{P}^0_{\mu\nu\alpha\beta} - \cfrac{1}{\varepsilon} \, \overline{\overline{P}}^0_{\mu\nu\alpha\beta} .
\end{align}
The inverse operator in arbitrary $d$ reads:
\begin{align}
    \left( \mathcal{O}_\text{GR+gf}^{-1} \right)_{\mu\nu\alpha\beta} = \cfrac{\varepsilon}{2} P^1_{\mu\nu\alpha\beta} + P^2_{\mu\nu\alpha\beta} + \cfrac{d-5}{d-2}\, P^0_{\mu\nu\alpha\beta} + \cfrac{(d-1)(\varepsilon-1)-\varepsilon}{d-2} \, \overline{P}^0_{\mu\nu\alpha\beta} - \cfrac{1}{d-2} \, \overline{\overline{P}}^0_{\mu\nu\alpha\beta}.
\end{align}
Consequently, the graviton propagator takes the following form:
\begin{align}
    \begin{gathered}
        \begin{fmffile}{Graviton_Propagator}
            \begin{fmfgraph*}(30,30)
                \fmfleft{L}
                \fmfright{R}
                \fmf{dbl_wiggly}{L,R}
                \fmflabel{$\mu\nu$}{L}
                \fmflabel{$\alpha\beta$}{R}
            \end{fmfgraph*}
        \end{fmffile}
    \end{gathered}
    \hspace{20pt } = \mathcal{G}_{\mu\nu\alpha\beta}(p) = i\,  \cfrac{\left( \mathcal{O}_\text{GR+gf}^{-1} \right)_{\mu\nu\alpha\beta}}{p^2} \,.
\end{align}
Finally, it takes a very simple and well-known for $\varepsilon = 2$:
\begin{align}\label{GR_propagator_simple}
    \mathcal{G}_{\mu\nu\alpha\beta}(p) = \cfrac12\, \left[ \eta_{\mu\alpha} \eta_{\nu\beta} + \eta_{\mu\beta} \eta_{\nu\alpha} - \cfrac{2}{d-2} \, \eta_{\mu\nu} \eta_{\alpha\beta} \right] \cfrac{i}{p^2} \, . \\ \nonumber
\end{align}

To conclude this section, let us discuss the relation between the given ghost constructions and the other known in the literature. The gauge fixing conditions \eqref{the_gauge_fixing_condition} are well-known in the literature and date back at least to \cite{Donoghue:1994dn,Donoghue:2017pgk}. The main difference between the gauge fixing condition \eqref{the_gauge_fixing_condition} and those used in \cite{Donoghue:1994dn,Donoghue:2017pgk} is in the following. One can introduce the standard harmonic gauge fixing term $g^{\mu\nu} \Gamma^\alpha_{\mu\nu}$ to the theory. Its contribution leading in $\kappa$ matches \eqref{the_gauge_fixing_condition}. However, the further one expands the gravity action in $\kappa$, the further they must expand the gauge fixing term. One of the previous papers devoted to FeynGrav followed exactly this approach \cite{Latosh:2023zsi}. It results in a theory with an infinite number of graviton-ghost vertices. More importantly, there is no reason to cut the expansion of the harmonic gauge fixing term at a given order in $\kappa$, for this goes against the perturbation theory logic.

In this section, we show that one can use the gauge fixing condition \eqref{the_gauge_fixing_condition} because it respects the sophisticated structure of gauge transformations in perturbative gravity. As we show above, the gauge transformation of the complete metric $g_{\mu\nu}$ consists of the gauge transformation of the background metric $\overline{g}_{\mu\nu}$ and gauge transformations of the perturbations $h_{\mu\nu}$. If one imposes gauge fixing conditions on the background metric, this would specify the physical frame in which we are working. On the contrary, gauge fixing conditions on perturbations do not in any way specify the physical frame, but rather remove redundancy in the description. Therefore, within the perturbative setting, the harmonic gauge fixing conditions $g^{\mu\nu} \Gamma^\alpha_{\mu\nu}$ imply that both perturbations and the physical frame should satisfy the given equation. On the contrary, equation \eqref{the_gauge_fixing_condition} means that in an arbitrary physical frame of our choosing, perturbations shall satisfy the given equation. And it is this separation of the two components of the gauge fixing condition that simplifies the theory's structure.

\subsection{Quadratic gravity}

The quadratic gravity is given by the following action, whose perturbative structure was discussed in \cite{Latosh:2024lhl}:
\begin{align}
    \mathcal{A}_\text{QG} = \int d^4 x \sqrt{-g} \left( - \cfrac{2}{\kappa^2} \right) \left[ R - \cfrac{1}{6\,m_0^2} \,R^2 +\cfrac{1}{m_2^2} \left( R_{\mu\nu}^2 -\cfrac13\, R^2 \right) \right].
\end{align}
The model contains the standard gravitons, scalars with mass $m_0$, and spin $2$ ghosts with mass $m_2$. The model is renormalisable, but it has ghost degrees of freedom, which challenge its unitarity \cite{Stelle:1976gc,Stelle:1977ry}. The discussion of unitarity and the applicability of this model lies far beyond the scope of this paper and will not be addressed. The goal of this paper (and of the \texttt{FeynGrav} package) is to provide tools to work with Feynman rules. Because of this, we treat the model perturbatively, as any other gravity model, and do not address any issues related to its unitarity or applicability.

Among all the states present in quadratic gravity, only gravitons contain gauge redundancy since the other perturbations are massive. Because of this, one can use exactly the same gauge fixing procedure that was used for general relativity. The procedure results in the same ghost propagator and ghost-graviton vertex. Graviton vertices remain unaffected, and the gauge fixing term only influences the graviton propagator. Similarly to general relativity, we introduce an operator describing the quadratic part of the action:
\begin{align}
    \begin{split}
        &\mathcal{A}_\text{QG} + \mathcal{A}_\text{gf} = \cfrac12\, \int\cfrac{d^4 p}{ (2\pi)^4} \, h^{\mu\nu} (-p) \, \left( \mathcal{O}_\text{QG+gf}\right)_{\mu\nu\alpha\beta} \square h^{\alpha\beta}(p) + \mathcal{O}\left( \kappa \right) ,\\
        & \left( \mathcal{O}_\text{QG+gf}\right)_{\mu\nu\alpha\beta}  = \cfrac{2}{\varepsilon} \, P^1_{\mu\nu\alpha\beta} +\left( 1 - \cfrac{p^2}{m_2^2} \right) P^2_{\mu\nu\alpha\beta} + \left( \cfrac{3}{\varepsilon} - 2 + 2 \,\cfrac{p^2}{m_0^2} \right) P^0_{\mu\nu\alpha\beta} + \cfrac{1}{\varepsilon} \, \overline{P}^0_{\mu\nu\alpha\beta} - \cfrac{1}{\varepsilon} \, \overline{\overline{P}}^0_{\mu\nu\alpha\beta} .
    \end{split}
\end{align}
Consequently, in arbitrary dimensions, the graviton propagator reads:
\begin{align}
    \begin{split}
        \mathcal{G}_{\mu\nu\alpha\beta}(p) = & \left( \mathcal{O}_\text{QG+gf}^{-1}\right)_{\mu\nu\alpha\beta} \,\cfrac{i}{p^2}  - \left( P^2_{\mu\nu\alpha\beta} + \cfrac{d-4}{d-1} \,P^0_{\mu\nu\alpha\beta} \right) \cfrac{i}{p^2 - m_2^2} \\
        & + \cfrac{1}{d-2} \left( \cfrac{3}{d-1} \,P^0_{\mu\nu\alpha\beta} + (d-1) \,\overline{P}^0_{\mu\nu\alpha\beta} + \overline{\overline{P}}^0_{\mu\nu\alpha\beta} \right) \,\cfrac{i}{p^2 - \cfrac{3(d-2)\,m_0^2 m_2^2}{(d-4)\, m_0^2 +2(d-1)\, m_2^2} }  \, .
    \end{split}
\end{align}

On the other hand, one can use the higher derivative gauge fixing and introduce derivative operators to the operator $Y^{\mu\nu}$. For the sake of simplicity, we use the same gauge fixing conditions $\mathcal{G}_\mu$ and introduce derivatives to the operator $Y^{\mu\nu}$ as follows:
\begin{align}
    Y^{\mu\nu} = \varepsilon_0 \, \overline{g}^{\mu\nu} \, \overline{\square} + \varepsilon_1 \overline{\nabla}^\mu \overline{\nabla}^\nu + \varepsilon_2 \overline{\nabla}^\nu \overline{\nabla}^\mu = \varepsilon_0 \, \overline{g}^{\mu\nu} \, \overline{\square} + \varepsilon^+ \left\{\overline\nabla^\mu,\overline\nabla^\nu\right\} + \varepsilon^- \left[ \overline\nabla^\mu,\overline\nabla^\nu \right] .
\end{align}
Here $\left\{\overline\nabla_\mu,\overline\nabla_\nu\right\} \overset{\text{def}}{=} \overline\nabla_\mu \overline\nabla_\nu + \overline\nabla_\nu \overline\nabla_\mu $ is the anticommutator, and the gauge fixing parameters $\varepsilon^\pm$ accompanying the symmetric and antisymmetric combinations of derivatives are related to $\varepsilon_{1,2}$ as follows:
\begin{align}
        \varepsilon^\pm & = \cfrac{\varepsilon_1 \pm \varepsilon_2}{2} \,.
\end{align}
Consequently, the part of the action describing the ghosts becomes
\begin{align}
    \begin{split}
        & \int d^4 x \sqrt{-\overline{g}}~ \overline{c}_\mu Y^{\mu\nu}  \left( \overline{\square} c_\nu + \overline{R}_{\nu\lambda} c^\lambda \right) \\
        & = \int d^4 x \sqrt{-\overline{g}} \,\Bigg[ \varepsilon_0 \left\{ \overline{c}_\mu \overline{\square}^2 c^\mu + \overline{c}_\mu \overline{\square} (\overline{R}^{\mu\nu} c_\nu ) \right\}  + \overline{c}_\mu \left[ \varepsilon^+  \{\overline{\nabla}^\mu,\overline{\nabla}^\nu \} - \varepsilon^- \overline{R}^{\mu\nu} \right](\overline{\square}c_\nu + \overline{R}_{\nu\lambda}c^\lambda) \Bigg].
    \end{split}
\end{align}
In the flat background spacetime limit, it reduces to the following form
\begin{align}
    \int\cfrac{d^4 p}{(2\pi)^4} \,c^\nu(-p) \,\varepsilon_0 \left[ \eta_{\mu\nu} + 2 \cfrac{\varepsilon^+}{\varepsilon_0} \, \cfrac{p_\mu p_\nu}{p^2} \right] \,p^4 c^\mu(p) .
\end{align}
It produces the following propagator:
\begin{align}
    \begin{gathered}
        \begin{fmffile}{Ghost_HD_Propagator}
            \begin{fmfgraph*}(30,30)
                \fmfleft{L}
                \fmfright{R}
                \fmf{dots_arrow}{L,R}
                \fmflabel{$\mu$}{L}
                \fmflabel{$\nu$}{R}
            \end{fmfgraph*}
        \end{fmffile}
    \end{gathered}
    \hspace{20pt} = i\, \cfrac{1}{\varepsilon_0} \left( \eta_{\mu\nu} - \cfrac{\varepsilon_1+\varepsilon_2}{\varepsilon_0+\varepsilon_1+\varepsilon_2} \, \cfrac{p_\mu \,p_\nu}{p^2} \right) \cfrac{1}{p^4} \, .
\end{align}

The gauge fixing term for the graviton takes the following form:
\begin{align}
    \begin{split}
        & \int d^4 x \sqrt{-\overline{g}} \left[ \cfrac{1}{2\,\varepsilon} \overline{\nabla}_\rho h_{\mu\nu} \left( \mathfrak{C}_\lambda{}^{\rho\mu\nu} \mathfrak{C}_\tau{}^{\sigma\alpha\beta} \,Y^{\lambda\tau} \right) \overline{\nabla}_\sigma h_{\alpha\beta} \right] \\
        &= -\cfrac{1}{2\,\varepsilon} \int d^4 x \sqrt{-\overline{g}} \,h_{\mu\nu} \,\mathfrak{C}^{\lambda\rho\mu\nu} \mathfrak{C}^{\tau\sigma\alpha\beta} \overline{\nabla}_\rho \left\{ \varepsilon_0 \, \overline{g}_{\lambda\tau} \overline{\square} + \varepsilon^+  \left\{  \overline{\nabla}_\lambda , \overline{\nabla}_\tau  \right\} \right\} \overline{\nabla}_\sigma h^{\alpha\beta} .
    \end{split}
\end{align}
The contribution with $\varepsilon^-$ does not enter this expression since it is an antisymmetric operator with respect to $\mu\leftrightarrow\nu$. In the flat background, the term has the following representation with the Nieuwenhuizen operators:
\begin{align}
    \begin{split}
        &-\cfrac{1}{2\,\varepsilon} \int\cfrac{d^4 p}{(2\pi)^4} \,h^{\mu\nu}(-p) \,\mathfrak{C}^{\lambda\rho\mu\nu} \mathfrak{C}^{\tau\sigma\alpha\beta} \left\{ \varepsilon_0 \, \eta_{\lambda\tau} \, p^2 \, p_\rho \, p_\sigma + 2\, \varepsilon^+  p_\lambda \, p_\tau \, p_\rho \, p_\sigma \right\} h^{\alpha\beta}  \\
        &= \cfrac{1}{2} \int \cfrac{d^4 p}{(2\pi)^4} \,h^{\mu\nu}(-p) \left(\mathcal{O}_\text{gfHD}\right)_{\mu\nu\alpha\beta} \, p^4 \, h^{\alpha\beta}(p),
    \end{split}
\end{align}
\begin{align}
    \left(\mathcal{O}_\text{gfHD}\right)_{\mu\nu\alpha\beta} = -\cfrac{1}{\varepsilon} \left( 2\, \varepsilon_0 P^1_{\mu\nu\alpha\beta} + \left( \varepsilon_0 + \varepsilon_1 + \varepsilon_2 \right)\left[ 3 P^0_{\mu\nu\alpha\beta} +  \overline{P}^0_{\mu\nu\alpha\beta} - \overline{\overline{P}}^0_{\mu\nu\alpha\beta} \right] \right) .
\end{align}
Consequently, the graviton propagator takes the following form:
\begin{align}
    \begin{split}
        \mathcal{G}_{\mu\nu\alpha\beta} =& \left[ P^2_{\mu\nu\alpha\beta} + \cfrac{d-5}{d-2} \, P^0_{\mu\nu\alpha\beta} - \cfrac{d-1}{d-2} \, \overline{P}^0_{\mu\nu\alpha\beta} - \cfrac{1}{d-2} \, \overline{\overline{P}}^0_{\mu\nu\alpha\beta} \right] \cfrac{i}{p^2}\\
        & - \left[ P^2_{\mu\nu\alpha\beta} + \cfrac{d-4}{d-1} \,P^0_{\mu\nu\alpha\beta} \right] \cfrac{i}{p^2 - m_2^2} - \varepsilon \left[ \cfrac{1}{2\varepsilon_0} \,P^1_{\mu\nu\alpha\beta} + \cfrac{1}{\varepsilon_0 + \varepsilon_1 + \varepsilon_2}\, \overline{P}^0_{\mu\nu\alpha\beta} \right] \cfrac{i}{p^4} \\
        & + \cfrac{1}{d-2}\left[ \cfrac{3}{d - 1} \, P^0_{\mu\nu\alpha\beta} + (d-1)\,\overline{P}^0_{\mu\nu\alpha\beta} + \overline{\overline{P}}^0_{\mu\nu\alpha\beta} \right] \cfrac{i}{p^2 - \cfrac{3(d-2)\,m_0^2 \,m_2^2}{(d-4) m_0^2+ 2(d-1) m_2^2} } \, .
    \end{split}
\end{align}

Lastly, the term describing the coupling of ghosts to gravitons reads
\begin{align}
        \mathcal{A}_{int} &= \cfrac{\kappa}{2} \int\! d^4 x \sqrt{-\overline{g}} ~ \overline{c}_\mu  \mathfrak{C}_\nu{}^{\tau\alpha\beta} \left[ \varepsilon_0 \, \overline{g}^{\mu\nu} \, \overline{\square} + \varepsilon_1 \overline\nabla^\mu\overline\nabla^\nu + \varepsilon_2 \overline\nabla^\nu \overline\nabla^\mu  \right]  \overline{\nabla}_\tau \left( h_{\alpha\mu} \overline{\nabla}_\beta c^\mu + h_{\beta\mu} \overline{\nabla}_\alpha c^\mu + \overline{\nabla}_\mu h_{\alpha\beta} c^\mu \right) .
\end{align}
For the flat background, it becomes
\begin{align}
    \begin{split}
        & \cfrac{\kappa}{2} \int \cfrac{p^4 p_1}{ (2\pi)^4} \,\cfrac{d^4 p_2}{ (2\pi)^4} \, \cfrac{d^4 k}{(2\pi)^4} \, (2\pi)^4 \delta\left(p_1+p_2+k\right) h^{\rho\sigma}(k)\\
        & \times \overline{c}^\nu(p_2) \mathfrak{C}^{\lambda\tau\alpha\beta} \left[ \varepsilon_0  \eta_{\nu\lambda}(p_1)^2 + (\varepsilon_1 \!+\! \varepsilon_2) (p_1)_\nu (p_1)_\lambda  \right] (p_1)_\tau \left( (p_1)_\alpha I_{\mu\beta\rho\sigma} \!+\! (p_1)_\beta I_{\mu\alpha\rho\sigma} \!+\! k_\mu I_{\alpha\beta\rho\sigma} \right)c^\mu(p_1).
    \end{split}
\end{align}
It produces the corresponding expression for the interaction vertex:
\begin{align}
\nonumber \\
\hspace{15pt}
    \begin{gathered}
        \begin{fmffile}{Ghost_Graviton_Vertex_HD}
            \begin{fmfgraph*}(35,35)
                \fmfleft{L}
                \fmfright{R1,R2}
                \fmf{dbl_wiggly}{L,V}
                \fmf{dots_arrow}{R1,V}
                \fmf{dots_arrow}{V,R2}
                \fmfdot{V}
                \fmflabel{$\rho\sigma,k$}{L}
                \fmflabel{$\mu,p_1$}{R1}
                \fmflabel{$\nu$}{R2}
            \end{fmfgraph*}
        \end{fmffile}
    \end{gathered}    
    =  i\,\cfrac{\kappa}{2} \, \mathfrak{C}^{\lambda\tau\alpha\beta} \left[ \varepsilon_0  \eta_{\nu\lambda}(p_1)^2 + (\varepsilon_1 \!+\! \varepsilon_2) (p_1)_\nu (p_1)_\lambda  \right] (p_1)_\tau \left( (p_1)_\alpha I_{\mu\beta\rho\sigma} \!+\! (p_1)_\beta I_{\mu\alpha\rho\sigma} \!+\! k_\mu I_{\alpha\beta\rho\sigma} \right) . \\ \nonumber
\end{align}

These formulae conclude the implementation of the BRST formalism for Feynman rules for gravity. They significantly simplify the overall structure of the Feynman rules since they only include a single graviton-ghost vertex and do not alter the structure of graviton vertices. In Sections \ref{Section_FG}, we discuss the implementation of these rules in \texttt{FeynGrav}.

\section{Cheung-Remmen Variables}\label{Section_CR}

This section introduces the Cheung-Remmen variables and describes the corresponding BRST complex \cite{Cheung:2017kzx,Alvarez:2025hym}. The result of this section is a specific parametrisation of the Hilbert action and a finite set of the corresponding Feynman rules.

The Cheung-Remmen variables $\gog^{\mu\nu}$ are the new dynamical variables that describe gravitational degrees of freedom. One introduces them by a non-linear mapping of an existing metric $g_{\mu\nu}$:
\begin{align}\label{Cheung-Remmen_variables}
    \gog^{\mu\nu} \overset{\text{def}}{=} \sqrt{-g} \, g^{\mu\nu} .
\end{align}
With the given definition, one can obtain an expression for the matrix inverse to the Cheung-Remmen variables:
\begin{align}\label{Cheung-Remmen_variables_inverse}
    \gog_{\mu\nu} = \cfrac{1}{\sqrt{-g}} ~ g_{\mu\nu} .
\end{align}

Equation \eqref{Cheung-Remmen_variables_inverse} is a conformal mapping of the metric for one specific choice of the conformal factor. Therefore, $\gog_{\mu\nu}$ is also a conformal metric. For simplicity, we will use notions of Cheung-Remmen variables and Cheung-Remmen metric interchangeably. Let us also note that, in contrast to the conventional approach, we treat the metric with upper indices, $\gog^{\mu\nu}$, as the fundamental variables.

To establish the form of the Hilbert action, we shall prove a series of theorems that describe the structure of various geometric quantities in terms of the Cheung-Remmen variables. First and foremost, it is useful to express the metric $g_{\mu\nu}$ and the corresponding volume factor $\sqrt{-g}$ through the Cheung-Remmen metric and its volume factor:
\begin{align}
        g_{\mu\nu} &= \left( \sqrt{-\gog} \right)^{- \frac{2}{d-2} } \gog_{\mu\nu} \, , & g^{\mu\nu} & = \left( \sqrt{-\gog} \right)^{\frac{2}{d-2}} \gog^{\mu\nu} \,, & \sqrt{-g} & = \left( \sqrt{-\gog} \right)^{- \frac{2}{d-2}}  \,.
\end{align}

Secondly, there is a relation between derivatives of the metric $g_{\mu\nu}$ and the Cheung-Remmen metric $\gog_{\mu\nu}$
\begin{align}
    \pd_\mu g_{\alpha \beta} & = \mathcal{C}_{\alpha\beta}{}^{\rho\sigma} \pd_\mu \gog_{\rho\sigma} , &  \pd_\mu \gog_{\alpha \beta} & = (\mathcal{C}^{-1})_{\alpha\beta}{}^{\rho\sigma} \pd_\mu g_{\rho\sigma} . 
\end{align}
Here, $\mathcal{C}$ and $\mathcal{C}^{-1}$ are the following tensors:
\begin{align}
    \mathcal{C}_{\alpha\beta}{}^{\rho\sigma} &= (\sqrt{-\gog})^{ -\frac{2}{d-2}} \left[ I_{\alpha\beta}{}^{\rho\sigma} - \cfrac{1}{d-2} \, \gog_{\alpha\beta}\,\gog^{\rho\sigma} \right] , &         (\mathcal{C}^{-1})_{\alpha\beta}{}^{\rho\sigma} &= (\sqrt{-\gog})^{ \frac{2}{d-2}} \left[ I_{\alpha\beta}{}^{\rho\sigma} - \cfrac12 \, \gog_{\alpha\beta}\,\gog^{\rho\sigma} \right] .
\end{align}
These tensors are inverse to each other:
\begin{align}
    \mathcal{C}_{\alpha\beta}{}^{\rho\sigma} (\mathcal{C}^{-1})_{\rho\sigma}{}^{\mu\nu} = (\mathcal{C}^{-1})_{\alpha\beta}{}^{\rho\sigma} \mathcal{C}_{\rho\sigma}{}^{\mu\nu} = I_{\alpha\beta}{}^{\mu\nu} .
\end{align}
    
This relation between metric derivatives allows us to obtain a relation between Christoffel symbols. It is more convenient to use the Christoffel symbols with all low indices because it is a linear combination of the metric derivatives:
\begin{align}
    \Gamma_{\alpha\mu\nu}[g] &\overset{\text{def}}{=} \cfrac12\, \left[ \pd_\mu g_{\nu\alpha} + \pd_\nu g_{\mu\alpha} - \pd_\alpha g_{\mu\nu} \right] .
\end{align}
Consequently, the following relation holds between the Christoffel symbols and metric derivatives:
\begin{align}
    \Gamma^{\alpha\mu\nu}[g] & = \Gamma_{\alpha\mu\nu}^{\lambda\rho\sigma} \, \pd_\lambda g_{\rho\sigma} \, , & \pd_\alpha g_{\mu\nu} &=  (\Gamma^{-1})_{\alpha\mu\nu}^{\lambda\rho\sigma} \Gamma_{\lambda\rho\sigma}[g] .
\end{align}
Here $\Gamma$ and $\Gamma^{-1}$ are tensors:
\begin{align}
    \Gamma_{\alpha\mu\nu}^{\lambda\rho\sigma} & = \cfrac12\left[ \delta^\lambda_\mu I_{\alpha\nu}{}^{\rho\sigma} + \delta^\lambda_\nu I_{\alpha\mu}{}^{\rho\sigma} - \delta^\lambda_\alpha I_{\mu\nu}{}^{\rho\sigma} \right] \, & (\Gamma^{-1})_{\alpha\mu\nu}^{\lambda\rho\sigma} &= \delta_\mu^\lambda I_{\alpha\nu}{}^{\rho\sigma} + \delta_\nu^\lambda I_{\alpha\mu}{}^{\rho\sigma} \, .
\end{align}
These tensors are inverse to each other:
\begin{align}
    \Gamma_{\alpha\mu\nu}^{\omega\rho\sigma}  (\Gamma^{-1})_{\omega\rho\sigma}^{\beta\lambda\tau} =(\Gamma^{-1})_{\alpha\mu\nu}^{\omega\rho\sigma}  \Gamma_{\omega\rho\sigma}^{\beta\lambda\tau} = \delta_\alpha^\beta \, I_{\mu\nu}{}^{\lambda\tau} \,.
\end{align}

The following theorem specifies the relation between Christoffel symbols calculated for $g_{\mu\nu}$ and $\gog_{\mu\nu}$. It is a direct corollary of the theorem above.

\begin{theorem}

    \begin{align}\label{Cheung-Remmen_Christoffel_mapping}
        \Gamma_{\alpha\mu\nu}[g] & =(\sqrt{-\gog})^{- \frac{2}{d-2} } \left[ \Gamma_{\alpha\mu\nu}[\gog] - \cfrac{2}{d-2}\left\{ \gog_{\alpha\mu} \Gamma^\lambda_{\nu\lambda}[\gog] + \gog_{\alpha\nu} \Gamma^\lambda_{\mu\lambda}[\gog] - \gog_{\mu\nu} \Gamma^\lambda_{\alpha\lambda}[\gog] \right\} \right] \,.
    \end{align}
    
\end{theorem}

Lastly, the following theorem specifies the relation between the Riemann tensors for metrics related by such a mapping. It is convenient to express such a relation in terms of the Kulkarni-Nomizu product.

\begin{theorem}

    {~}

    \noindent The following relation between the Riemann tensors holds:
    \begin{align}\label{Cheung-Remmen_Riemann_mapping}
        R_{\mu\nu\alpha\beta}[g]=(\sqrt{-\gog})^{- \frac{2}{d-2} } \left[ R_{\mu\nu\alpha\beta}[\gog] + \left( \gog \owedge P[\gog] \right)_{\mu\nu\alpha\beta} \right] .
    \end{align}
    Here $\owedge$ denotes the Kulkarni-Nomizu product
    \begin{align}
        (X \owedge Y)_{\mu\nu\alpha\beta} \overset{\text{def}}{=} X_{\mu\alpha} Y_{\nu\beta} + X_{\nu\beta} Y_{\mu\alpha} - X_{\mu\beta} Y_{\nu\alpha} - X_{\nu\alpha} Y_{\mu\beta}.
    \end{align}
    In our case, the Kulkarni-Nomizu product is constructed from the Cheung-Remmen metric $\gog$ and $P[\gog]$:
    \begin{align}
        \begin{split}
                P_{\mu\nu}[\gog] &\overset{\text{def}}{=}\cfrac{1}{d-2} \nabla_\nu \omega_\mu + \cfrac{1}{(d-2)^2} \, \left[ \omega_\mu \omega_\nu - \cfrac12\, \gog_{\mu\nu} \omega^2 \right] ,\\
                \omega_\mu &\overset{\text{def}}{=} \pd_\mu \log \sqrt{-\gog} \, .
        \end{split}
    \end{align}
    
\end{theorem}

The following comment on the derivation is due. Most often, one considers a conformal mapping of the following form:
\begin{align}
    g_{\mu\nu} \to e^{2\varphi} \, g_{\mu\nu} \, .
\end{align}
In this expression, $\varphi$ is assumed to be a scalar. This setup results in the following formula for the Riemann tensor:
\begin{align}
    R_{\mu\nu\alpha\beta} \to e^{2\varphi} \left[ R_{\mu\nu\alpha\beta} + g\owedge T[\varphi] \right] \, ,
\end{align}
where 
\begin{align}
    T[\varphi]_{\mu\nu} \overset{\text{def}}{=} \nabla_{\mu\nu} \varphi - \nabla_\mu \varphi \nabla_\nu \varphi + \cfrac12 \, g_{\mu\nu} \left\lvert\pd\varphi\right\rvert^2  \,.
\end{align}
This formula only holds for $\varphi$ that are scalar, so one must not use it for \eqref{Cheung-Remmen_variables_inverse}. The factor $\sqrt{-\gog}$ present in \eqref{Cheung-Remmen_variables_inverse} is not a scalar, but a tensor density. Therefore, one must use \eqref{Cheung-Remmen_Christoffel_mapping} to derive \eqref{Cheung-Remmen_Riemann_mapping}. The derivation itself reduces to careful index manipulation.

The following theorem is a direct corollary of \eqref{Cheung-Remmen_Riemann_mapping}. It is particularly interesting because it shows that in the Cheung-Remmen variables the elementary volume $\sqrt{-g}$ does not enter the Hilbert Lagrangian. This is the main reason why one can reduce the Hilbert action to the polynomial form.

\begin{theorem}

    \begin{align}
        \begin{split}
            R_{\mu\nu}[g] &= R_{\mu\nu}[\gog] + (d-2) P_{\mu\nu}[\gog] + \gog_{\mu\nu} \gog^{\alpha\beta} P_{\alpha\beta}[\gog] \, , \\
            R[g] &= (\sqrt{-\gog})^{\frac{2}{d-2}} \, \left( R[\gog ] + 2 (d-1) \gog^{\alpha\beta} P_{\alpha\beta}[\gog] \right) \, , \\
            \sqrt{-g} R[g] &= R[\gog ] + 2 (d-1) \gog^{\alpha\beta} P_{\alpha\beta}[\gog] \,.
        \end{split}
    \end{align}

\end{theorem}

To proceed with mapping the Hilbert action, we will employ further simplification. In the original formulation, the action contains the second-order derivative of the metric. One can always remove them with an integration by parts:
\begin{align}
    \int d^d x \, \sqrt{-g} \, R = \int d^d x \, \sqrt{-g} \, g^{\mu\nu} g^{\alpha\beta} g^{\rho\sigma} \left[ \Gamma_{\alpha\mu\rho} \Gamma_{\sigma\nu\beta} - \Gamma_{\alpha\mu\nu} \Gamma_{\rho\beta\sigma} \right] + \text{surface terms}.
\end{align}
Further, we omit all surface terms for simplicity. The following theorem specifies the structure of the Hilbert action modulo surface terms.

\begin{theorem}

    \begin{align}
        \begin{split}
            \int d^d x \, \sqrt{-g} \, R &= \int d^d x ~ \cfrac{1}{2} ~ \pd_{\lambda_1} \gog_{\alpha_1\beta_1} \mathcal{O}^{\lambda_1\alpha_1\beta_1\lambda_2\alpha_2\beta_2}(\gog) \pd_{\lambda_2} \gog_{\alpha_2\beta_2}  \, , \\
            \mathcal{O}^{\lambda_1\alpha_1\beta_1\lambda_2\alpha_2\beta_2}(\gog) & =  \cfrac{1}{4} \left( \gog^{\lambda_1\alpha_2} \gog^{\lambda_2\alpha_1}\gog^{\beta_1\beta_2} + \gog^{\lambda_1\beta_2} \gog^{\lambda_2\alpha_1}\gog^{\beta_1\alpha_2} + \gog^{\lambda_1\alpha_2} \gog^{\lambda_2\beta_1}\gog^{\alpha_1\beta_2} + \gog^{\lambda_1\beta_2} \gog^{\lambda_2\beta_1}\gog^{\alpha_1\alpha_2} \right) \\
            & - \cfrac14 \,\gog^{\lambda_1\lambda_2} \left( \gog^{\alpha_1\alpha_2} \gog^{\beta_1\beta_2} + \gog^{\alpha_1\beta_2} \gog^{\alpha_2\beta_1} \right) + \cfrac12\, \cfrac{1}{d-2} \, \gog^{\lambda_1\lambda_2} \gog^{\alpha_1\beta_1}\gog^{\alpha_2\beta_2}  .
        \end{split}
    \end{align}
    
\end{theorem}

\noindent The operator $\mathcal{O}^{\lambda_1\alpha_1\beta_1\lambda_2\alpha_2\beta_2}$ depends on the metric $\gog$ and admits the following symmetries by the construction:
\begin{align}
    \mathcal{O}^{\lambda_1\alpha_1\beta_1\lambda_2\alpha_2\beta_2} = \mathcal{O}^{\lambda_1\beta_1\alpha_1\lambda_2\alpha_2\beta_2} = \mathcal{O}^{\lambda_1\alpha_1\beta_1\lambda_2\beta_2\alpha_2}= \mathcal{O}^{\lambda_2\alpha_2\beta_2\lambda_1\alpha_1\beta_1} .
\end{align}
The previous papers \cite{Cheung:2017kzx,Alvarez:2025hym} used the following operator to describe this structure
\begin{align}
    \gog^{\lambda_1\alpha_2} \gog^{\lambda_2\alpha_1}\gog^{\beta_1\beta_2} - \cfrac12 \,\gog^{\lambda_1\lambda_2} \gog^{\alpha_1\alpha_2} \gog^{\beta_1\beta_2} + \cfrac12\, \cfrac{1}{d-2} \, \gog^{\lambda_1\lambda_2} \gog^{\alpha_1\beta_1}\gog^{\alpha_2\beta_2} .
\end{align}
One can obtain $ \mathcal{O}^{\lambda_1\alpha_1\beta_1\lambda_2\alpha_2\beta_2}$ by symmetrising the expression above.

The previous theorem shows that not only is the action polynomial in fields, but also the derivatives act only on the Cheung-Remmen metric with the low indices $\gog_{\mu\nu}$. It is this feature of the action that helps us to introduce an auxiliary variable. The following theorem specifies the method for introducing an auxiliary variable to construct an equivalent action.

\begin{theorem}

    {~}

    \noindent The following actions $\mathcal{S}_1$ and $\mathcal{S}_2$ are equivalent:
    \begin{align}
        \begin{split}
            \mathcal{S}_1 [\gog] & = \int d^d x \left[ \cfrac12\, \pd_{\lambda_1}\gog_{\alpha_1\beta_1} \mathcal{O}^{\lambda_1\alpha_1\beta_1\lambda_2\alpha_2\beta_2}(\gog) \, \pd_{\lambda_2} \gog_{\alpha_2\beta_2}  \right] , \\
            \mathcal{S}_2 [\gog,A] & = \int d^d x \left[ -\cfrac12 \, A^{\lambda_1\alpha_1\beta_1} (\mathcal{O}^{-1})_{\lambda_1\alpha_1\beta_1\lambda_2\alpha_2\beta_2} (\gog) \,A^{\lambda_2\alpha_2\beta_2} + A^{\lambda\alpha\beta} \pd_\lambda \gog_{\alpha\beta} \right] .
        \end{split}
    \end{align}
    Here, operator $\mathcal{O}^{-1}$ is an inverse of the $\mathcal{O}$:
    \begin{align}
        (\mathcal{O}^{-1})_{\lambda_1\alpha_1\beta_1\lambda_2\alpha_2\beta_2} (\gog) \mathcal{O}^{\lambda_2\alpha_2\beta_2\lambda_3\alpha_3\beta_3} (\gog) = \delta_{\lambda_1}^{\lambda_3} \, I_{\alpha_1\beta_1}{}^{\alpha_3\beta_3} .
    \end{align}
    
\end{theorem}

\begin{proof}

    To show the equivalence of these actions, one shall obtain the field equations for $A$:
    \begin{align}
        -(\mathcal{O}^{-1})_{\lambda\alpha\beta\lambda_1\alpha_1\beta_1}(\gog) A^{\lambda_1\alpha_1\beta_1} + \pd_\lambda \gog_{\alpha\beta} = 0 \, .
    \end{align}
    Because operator $\mathcal{O}$ is invertible, the field equations reduce to an algebraic relation:
    \begin{align}
        A^{\lambda\alpha\beta} = \mathcal{O}^{\lambda\alpha\beta\lambda_1\alpha_1\beta_1}(\gog) \pd_{\lambda_1} \gog_{\alpha_1\beta_1} .
    \end{align}
    When this relation is satisfy, the action $\mathcal{S}_2$ matches $\mathcal{S}_1$ exactly:
    \begin{align}
        \mathcal{S}_2 [\gog_{\mu\nu},\mathcal{O}^{\lambda\alpha\beta\lambda_1\alpha_1\beta_1}(\gog) \pd_{\lambda_1} \gog_{\alpha_1\beta_1}] = \mathcal{S}_1[\gog_{\mu\nu}].
    \end{align}
    
\end{proof}

\noindent Lastly, one shall only fine the operator $\mathcal{O}^{-1}$ to complete the mapping of the Hilbert action.

\begin{theorem}

    {~}

    \noindent The operator $\mathcal{O}^{\lambda_1\alpha_1\beta_1\lambda_2\alpha_2\beta_2}$ is invertible. The inverse operator reads:
    \begin{align}
        \begin{split}
            \left(\mathcal{O}^{-1}\right)_{\lambda_1\alpha_1\beta_1\lambda_2\alpha_2\beta_2} =& \cfrac12\left( \gog_{\lambda_1\alpha_2} \gog_{\lambda_2\alpha_1} \gog_{\beta_1\beta_2} \!+\! \gog_{\lambda_1\beta_2} \gog_{\lambda_2\alpha_1} \gog_{\beta_1\alpha_2} \!+\! \gog_{\lambda_1\alpha_2} \gog_{\lambda_2\beta_1} \gog_{\alpha_1\beta_2} \!+\! \gog_{\lambda_1\beta_2} \gog_{\lambda_2\beta_1} \gog_{\alpha_1\alpha_2} \right) \\
            & - \cfrac12\, \cfrac{1}{d-1} \left( \gog_{\lambda_1\alpha_1} \gog_{\lambda_2\alpha_2} \gog_{\beta_1\beta_2} \!+\! \gog_{\lambda_1\alpha_1} \gog_{\lambda_2\beta_2} \gog_{\beta_1\alpha_2} \!+\! \gog_{\lambda_1\beta_1} \gog_{\lambda_2\alpha_2} \gog_{\alpha_1\beta_2} \!+\! \gog_{\lambda_1\beta_1} \gog_{\lambda_2\beta_2} \gog_{\alpha_1\alpha_2} \right).
        \end{split}
    \end{align}
    
\end{theorem}

\noindent Similarly to the previous theorem, the operator $ \mathcal{O}^{-1}_{\lambda_1\alpha_1\beta_1\lambda_2\alpha_2\beta_2}$ admits the following symmetry by construction:
\begin{align}
     \left(\mathcal{O}^{-1}\right)_{\lambda_1\alpha_1\beta_1\lambda_2\alpha_2\beta_2} =  \left(\mathcal{O}^{-1}\right)_{\lambda_1\beta_1\alpha_1\lambda_2\alpha_2\beta_2} =  \left(\mathcal{O}^{-1}\right)_{\lambda_1\alpha_1\beta_1\lambda_2\beta_2\alpha_2} =  \left(\mathcal{O}^{-1}\right)_{\lambda_2\alpha_2\beta_2\lambda_1\alpha_1\beta_1} .
\end{align}
One can construct it by a suitable symmetrisation of the following expression:
\begin{align}
    2\, \left( \gog_{\lambda_1\alpha_2} \gog_{\lambda_2\alpha_1} - \cfrac{1}{d-1} \, \gog_{\lambda_1\alpha_1} \gog_{\lambda_2\alpha_2} \right) \gog_{\beta_1\beta_2} .
\end{align}

The following theorem applies the results of all the above theorems to map the Hilbert action on the Cheung-Remmen variables.

\begin{theorem}[Cheung-Remmen action]

    {~}

    \noindent The Hilbert action in $d$ dimensions
    \begin{align}
        S_\text{H}[g] = -\cfrac{2}{ \kappa^{d-2} } \int d^d x ~\sqrt{-g} \, R
    \end{align}
    is equivalent to the following Cheung-Remmen action:
    \begin{align}\label{Cheung-Remmen_action}
        \begin{split}
            S_\text{CR}[\gog,A] &= - \cfrac{2}{ \kappa^{d-2} }\int d^d x \left[ -\cfrac12\, A^{\lambda_1\mu_1\nu_1}  \left(\mathcal{O}^{-1}\right)_{\lambda_1\mu_1\nu_1\lambda_2\mu_2\nu_2} (\gog)  A^{\lambda_2\mu_2\nu_2} + A^{\lambda\mu\nu}\pd_\lambda \gog_{\mu\nu} \right] \\
            &= - \cfrac{2}{ \kappa^{d-2} }\int d^d x \left[ A^{\lambda_1}_{\mu_1\nu_1} \,A^{\lambda_2}_{\mu_2\nu_2} \left( \delta_{\lambda_1}^{\mu_2} \delta_{\lambda_2}^{\mu_1} - \cfrac{1}{d-1} \, \delta_{\lambda_1}^{\mu_1} \, \delta_{\lambda_2}^{\mu_2} \right) \gog^{\nu_1\nu_2} - A^\lambda_{\mu\nu} \, \pd_\lambda \gog^{\mu\nu} \right].
        \end{split}
    \end{align}
    
\end{theorem}

\noindent We present the second line of the theorem solely to highlight the direct relation with the previous publication \cite{Cheung:2017kzx,Alvarez:2025hym}. For further calculations, it is useful to work with the symmetric operators $\mathcal{O}$ and $\mathcal{O}^{-1}$.

The action \eqref{Cheung-Remmen_action} is applicable in the most general setup. Our goal is to derive the Feynman rules, so we shall expand both the Cheung-Remmen and auxiliary variables about the flat background. To work with $A^\lambda_{\mu\nu}$ variables, we shall manipulate the indices of these operators in the following way:
\begin{align}
    \begin{split}
        \mathcal{O}{}^{\lambda_1} {}_{\alpha_1\beta_1} {}^{\lambda_2}{}_{\alpha_2\beta_2} (\gog) & \overset{\text{def}}{=} \mathcal{O}^{\lambda_1\mu_1\nu_1\lambda_2\mu_2\nu_2} \, \gog_{\mu_1\alpha_1} \, \gog_{\nu_1\beta_1} \, \gog_{\mu_2\alpha_2} \, \gog_{\nu_2\beta_2} \\
        & = \cfrac14 \left( \delta^{\lambda_1}_{\alpha_2} \delta^{\lambda_2}_{\alpha_1}\gog_{\beta_1\beta_2} + \cdots \right) - \cfrac14 \, \gog^{\lambda_1\lambda_2} \left( \gog_{\alpha_1\alpha_2}\gog_{\beta_1\beta_2} + \cdots \right) + \cfrac12\, \cfrac{1}{d-2} \, \gog^{\lambda_1\lambda_2}\,\gog_{\alpha_1\beta_1}\,\gog_{\alpha_2\beta_2} \, ,\\
        \left(\mathcal{O}^{-1}\right) {}_{\lambda_1}{}^{\alpha_1\beta_1}{}_{\lambda_2}{}^{\alpha_2\beta_2} (\gog) & \overset{\text{def}}{=} \left(\mathcal{O}^{-1}\right)_{\lambda_1\mu_1\nu_1\lambda_2\mu_2\nu_2} \, \gog^{\mu_1\alpha_1} \, \gog^{\nu_1\beta_1} \, \gog^{\mu_2\alpha_2} \, \gog^{\nu_2\beta_2} \\
        & = \cfrac12\left(\delta_{\lambda_1}^{\alpha_2} \delta_{\lambda_2}^{\alpha_1} \, \gog^{\beta_1\beta_2} + \cdots \right) - \cfrac12 \, \cfrac{1}{d-1}\left( \delta_{\lambda_1}^{\alpha_1}\delta_{\lambda_2}^{\alpha_2} \gog^{\beta_1\beta_2} + \cdots\right) \,.
    \end{split}
\end{align}
The operator $\mathcal{O}^{-1}$ which enters the Cheung-Remmen action \eqref{Cheung-Remmen_action} is linear with respect to $\gog^{\mu\nu}$. Consequently, its perturbative expansion will contain only two terms, which result in the simpler structure of the Feynman rules.

We introduce small perturbations $\goh^{\mu\nu}$ of the flat background metric $\eta^{\mu\nu}$ in the following way:
\begin{align}\label{CH_perturbative_expansion}
    \gog^{\mu\nu} = \eta^{\mu\nu} - \kappa \, \goh^{\mu\nu} .
\end{align}
The variables $\goh^{\mu\nu}$ admit a non-linear, but algebraic relation with the standard perturbative variables:
\begin{align}
    \goh^{\mu\nu} = \sum\limits_{n=0}^{\infty} (-1)^{n+1} \kappa^{n-1} (h^n)^{\mu\nu} .
\end{align}
The immediate corollary of such a choice of perturbative variables is the following split of the $\mathcal{O}^{-1}$ operator:
\begin{align}
    \begin{split}
        \left(\mathcal{O}^{-1}\right) {}_{\lambda_1}{}^{\alpha_1\beta_1}{}_{\lambda_2}{}^{\alpha_2\beta_2} (\eta^{\mu\nu} - \kappa \, \goh^{\mu\nu}) &= \left(\mathcal{O}^{-1}\right) {}_{\lambda_1}{}^{\alpha_1\beta_1}{}_{\lambda_2}{}^{\alpha_2\beta_2} (\eta^{\mu\nu}) - \kappa \left(\mathcal{O}^{-1}\right) {}_{\lambda_1}{}^{\alpha_1\beta_1}{}_{\lambda_2}{}^{\alpha_2\beta_2} (\goh^{\mu\nu}) \, , \\
        \left(\mathcal{O}^{-1}\right) {}_{\lambda_1}{}^{\alpha_1\beta_1}{}_{\lambda_2}{}^{\alpha_2\beta_2} (\eta) & = \cfrac12\left(\delta_{\lambda_1}^{\alpha_2} \delta_{\lambda_2}^{\alpha_1} \, \eta^{\beta_1\beta_2} + \cdots \right) - \cfrac12 \, \cfrac{1}{d-1}\left( \delta_{\lambda_1}^{\alpha_1}\delta_{\lambda_2}^{\alpha_2} \eta^{\beta_1\beta_2} + \cdots\right) \, , \\
        \left(\mathcal{O}^{-1}\right) {}_{\lambda_1}{}^{\alpha_1\beta_1}{}_{\lambda_2}{}^{\alpha_2\beta_2} (\goh) & = \cfrac12\left(\delta_{\lambda_1}^{\alpha_2} \delta_{\lambda_2}^{\alpha_1} \, \goh^{\beta_1\beta_2} + \cdots \right) - \cfrac12 \, \cfrac{1}{d-1}\left( \delta_{\lambda_1}^{\alpha_1}\delta_{\lambda_2}^{\alpha_2} \goh^{\beta_1\beta_2} + \cdots\right) .
     \end{split}
\end{align}

The choice of the perturbative expansion \eqref{CH_perturbative_expansion} uniquely fixes the perturbative expansion for the auxiliary variable $A^\lambda_{\mu\nu}$. The Cheung-Remmen action \eqref{Cheung-Remmen_action} contains a mixing term between variables $A$ and $\gog$. In order to exclude this term, one must introduce a term proportional to $\pd\goh$ to the perturbative expansion of $A$.

\begin{theorem}[Cheung-Remmen action in perturbative variables]

    {~}

    \noindent If one defines the perturbative variables $\goh$ and $B$ in the following way:
    \begin{align}
        \begin{split}
            \gog^{\mu\nu} &= \eta^{\mu\nu} - \kappa \, \goh^{\mu\nu} , \\
            A^\lambda_{\mu\nu} &= B^\lambda_{\mu\nu} + \kappa \, \mathcal{O}^\lambda{}_{\mu\nu}{}^\alpha{}_{\rho\sigma} (\eta) ~ \pd_\alpha \goh^{\rho\sigma}  .
        \end{split}
    \end{align}
    Then the action \eqref{Cheung-Remmen_action} splits in the following parts:
    \begin{align}
        S_\text{CR}[\goh,B] &= S_{\goh\goh} + S_{BB} + S_{\goh\goh\goh} + S_{B\goh\goh} + S_{BB\goh} \, ,
    \end{align}
    \begin{align}
        \begin{split}
            S_{\goh\goh} & = \int d^d x\left[ \cfrac12\, \pd_{\lambda_1} \goh^{\mu_1\nu_1} \left\{ -2\,\kappa^{4-d} \, \mathcal{O}^{\lambda_1}{}_{\mu_2\nu_2}{}^{\lambda_2}{}_{\mu_2\nu_2} (\eta) \right\} \pd_{\lambda_2} \goh^{\mu_2\nu_2} \right] \\
            & = \int d^d x \, \kappa^{4-d}\left[ \goh^{\mu\nu} \pd_\lambda\pd_\mu \goh^\lambda{}_\nu + \cfrac12 \, \goh^{\mu\nu} \square \goh_{\mu\nu} - \cfrac12\,\cfrac{1}{d-2}\,\goh\square \, \goh \right] ,
        \end{split}
    \end{align}
    \begin{align}
        \begin{split}
            S_{BB} & =\int d^d x \left[ \cfrac12 \, B^{\lambda_1}_{\mu_1\nu_1} \left\{ 2\, \kappa^{2-d} (\mathcal {O}^{-1})_{\lambda_1}{}^{\mu_1\nu_1}{}_{\lambda_2}{}^{\mu_2\nu_2} (\eta) \right\} \, B^{\lambda_2}_{\mu_2\nu_2} \right] \\
            & = \int d^d x \, 2 \, \kappa^{2-d} \left[ B^{\mu\nu\lambda} B_{\nu\mu\lambda} - \cfrac{1}{d-1} \, B^\mu{}_{\mu\lambda} B^\nu{}_\nu{}^\lambda\right] ,
        \end{split}
    \end{align}
    \begin{align}
        \begin{split}
            S_{\goh\goh\goh} & = \int d^4 x \, \pd_{\lambda_1}\goh^{\mu_1\nu_1}\left\{ - \kappa^{5-d} \,  \mathcal{O}^{\lambda_1}{}_{\mu_1\nu_1}{}^{\alpha_1}{}_{\rho_1\sigma_1}(\eta) \left(\mathcal{O}^{-1}\right)_{\alpha_1}{}^{\rho_1\sigma_1}{}_{\alpha_2}{}^{\rho_2\sigma_2}(\goh) \mathcal{O}^{\alpha_2}{}_{\rho_2\sigma_2}{}^{\alpha_3}{}_{\rho_3\sigma_3}(\eta) \right\} \pd_{\lambda_2} \goh^{\mu_2\nu_2} \\
            & = \int d^d x \, \kappa^{5-d} \, \goh^{\mu\nu} \left[ -\cfrac12\, \pd_\mu \goh_{\alpha\beta} \pd_\nu \goh^{\alpha\beta} - \pd_\alpha \goh_{\mu\beta} \left( \pd^\beta \goh_\nu{}^\alpha - \pd^\alpha \goh_\nu{}^{\beta}\right) + \cfrac{1}{d-2}\left( \cfrac12\,  \pd_\mu \goh \pd_\nu\goh - \pd_\lambda \goh_{\mu\nu} \pd^\lambda \goh \right) \right] , 
        \end{split}
    \end{align}
    \begin{align}
        \begin{split}
            S_{B\goh\goh} & =\int d^d x \, B^{\lambda_1}_{\mu_1\nu_1} \left\{ -2 \, \kappa^{4-d} \, \left(O^{-1}\right)_{\lambda_1}{}^{\mu_1\nu_1}{}_{\alpha}{}^{\rho\sigma}(\goh)  \mathcal{O}^\alpha{}_{\mu\nu}{}^{\lambda_2}{}_{\mu_2\nu_2} (\eta) \right\} \pd_{\lambda_2} \goh^{\mu_2\nu_2} \\
            & = \int d^d x \, \kappa^{4-d} \left[ \cfrac{2}{d-2} \left( B^\alpha{}_{\alpha\beta} \goh^{\beta\mu} \pd_\mu \goh - B^{\alpha\beta\mu}\goh_{\alpha\beta}\pd_\mu \goh \right) -2 B^{\nu\beta\lambda} \goh_\beta{}^\mu \left( \pd_\mu \goh_{\nu\lambda} + \pd_\nu \goh_{\mu\lambda} - \pd_\lambda \goh_{\mu\nu} \right) \right] ,
        \end{split}
    \end{align}
    \begin{align}
        \begin{split}
            S_{BB\goh} & = \int d^d x \, B^{\lambda_1}_{\mu_1\nu_1}\left\{ - \kappa^{3-d} \, (\mathcal{O}^{-1})_{\lambda_1}{}^{\mu_1\nu_1}{}_{\lambda_2}{}^{\mu_2\nu_2} (\goh) \right\} B^{\lambda_2}_{\mu_2\nu_2} \\
            & = \int d^d x \, \kappa^{3-d} \, \goh^{\mu\nu} \left[ -2 \, B^{\alpha\beta}{}_\mu B_{\beta\alpha\nu} + \cfrac{2}{d-1}\,B^\alpha_{\alpha\mu} B^\beta{}_{\beta\nu} \right] .
        \end{split}
    \end{align}

\end{theorem}

Before proceeding, we shall comment on the choice of the background metric. One may expect that the same calculations may be done with a generic background $\overline{g}$ with little to no changes. From the calculation point of view, this is true. On the contrary, from the physical point of view, the choice of a generic background introduces a source term for the $B$ variable. Namely, the action contains a term $A^{\lambda}_{\mu\nu} \pd_\lambda \gog^{\mu\nu}$ which will produce a term $B^{\lambda}_{\mu\nu} \pd_\lambda \overline{g}^{\mu\nu}$ which is a linear coupling between $B$ and the background. Consequently, the choice of the flat background is determined if one wants to exclude such a source term.

Lastly, we shall bring the action to a form suitable for the functional integral.

\begin{theorem}

    {~}

    \noindent The action \eqref{CH_perturbative_expansion} takes the following form that is more suitable for the treatment with the functional integral.
    \begin{align}\label{Cheung-Remmen_action_perturbative}
        \mathcal{A}_\text{CR}[\goh,B] = \mathcal{A}_{\goh\goh} + \mathcal{A}_{BB} + \mathcal{A}_{\goh\goh\goh} + \mathcal{A}_{B\goh\goh} + \mathcal{A}_{BB\goh} \, ,
    \end{align}
    Terms $\mathcal{A}_{\goh\goh}$ and $\mathcal{A}_{BB}$ are quadratic in fields:
    \begin{align}
        \begin{split}
            \mathcal{A}_{\goh\goh} &= \int \cfrac{d^d p}{(2\pi)^d} ~ \goh^{\mu_1\nu_1}(-p) \, \goh^{\mu_2\nu_2}(p) \, \cfrac12\, \mathcal{P}_{\mu_1\nu_1\mu_2\nu_2} \, ,\\
            \mathcal{A}_\text{BB} &= \int \cfrac{d^d p}{(2\pi)^d} ~ B^{\lambda_1}_{\mu_1\nu_1}(-p) \, B^{\lambda_1}_{\mu_2\nu_2}(p) \, \cfrac12\, \left(\mathcal{P}_\text{B}\right)_{\lambda_1}{}^{\mu_1\nu_1}{}_{\lambda_2}{}^{\mu_2\nu_2} \,.
        \end{split}
    \end{align}
    The operators $\mathcal{P}$ and $\mathcal{P}_B$ take the following form:
    \begin{align}
        \begin{split}
            &\mathcal{P}_{\mu_1\nu_1\mu_2\nu_2} = \kappa^{4-d}\left[ -\cfrac12\left( p_{\mu_1} p_{\mu_2} \eta_{\nu_1\nu_2} + \cdots \right) + p^2\left( \cfrac12\left(\eta_{\mu_1\mu_2}\eta_{\nu_1\nu_2} + \eta_{\mu_1\nu_2}\eta_{\mu_2\nu_1}\right)  -\cfrac{1}{d-2} \, \eta_{\mu_1\nu_1}\eta_{\mu_2\nu_2} \right) \right] \, , \\
            & \left(\mathcal{P}_\text{B}\right)_{\lambda_1}{}^{\mu_1\nu_1}{}_{\lambda_2}{}^{\mu_2\nu_2} =\kappa^{2-d}\left[ \left(\delta_{\lambda_1}^{\mu_2} \delta_{\lambda_2}^{\mu_1} \, \eta^{\nu_1\nu_2} + \cdots \right) -  \cfrac{1}{d-1}\left( \delta_{\lambda_1}^{\mu_1}\delta_{\lambda_2}^{\mu_2} \eta^{\nu_1\nu_2} + \cdots\right) \right] \, ,
        \end{split}
    \end{align}
    The other terms describe the interaction of fields $\goh$ and $B$. We present them in the following symmetric form:
    \begin{align}
        \begin{split}
            \mathcal{A}_{\goh\goh\goh} & = \!\!\!\int\!\! \cfrac{d^d p_1}{(2\pi)^d} \cfrac{d^d p_2}{(2\pi)^d} \cfrac{d^d p_3}{(2\pi)^d} \, \goh^{\mu_1\nu_1}(p_1) \goh^{\mu_2\nu_2}(p_2)\goh^{\mu_3\nu_3}(p_3)\, (2\pi)^d \delta\left( p_1 \!+\! p_2 \!+\! p_3\right)  \cfrac{1}{3!} \,\mathcal{V}_{\mu_1\nu_1\mu_2\nu_2\mu_3\nu_3}(p_1,p_2,p_3) \,, \\
            \mathcal{A}_{B\goh\goh} & = \!\!\!\int\!\! \cfrac{d^d q}{(2\pi)^d} \cfrac{d^d p_1}{(2\pi)^d} \cfrac{d^d p_2}{(2\pi)^d} \, B^{\alpha}_{\rho\sigma}(q) \goh^{\mu_1\nu_1}(p_1) \goh^{\mu_2\nu_2}(p_2)\, (2\pi)^d \delta\left( q \!+\! p_1 \!+\! p_2\right)  \cfrac{1}{2!} \,\mathcal{V}_\alpha{}^{\rho\sigma}{}_{\mu_1\nu_1\mu_2\nu_2}(p_1,p_2) \,, \\
            \mathcal{A}_{BB\goh} & = \!\!\!\int\!\! \cfrac{d^d q_1}{(2\pi)^d} \cfrac{d^d q_2}{(2\pi)^d} \cfrac{d^d p}{(2\pi)^d} \, B^{\alpha_1}_{\rho_1\sigma_1}(q_1) B^{\alpha_2}_{\rho_2\sigma_2}(q_2) \goh^{\mu\nu}(p) \, (2\pi)^d \delta\left( q_1 \!+\! q_2 \!+\! p \right) \cfrac{1}{2!} \,\mathcal{V}_{\alpha_1}{}^{\rho_1\sigma_1}{}_{\alpha_2}{}^{\rho_2\sigma_2}{}_{\mu\nu} \,.
        \end{split}
    \end{align}
    Factors $1/3!$ and $1/2!$ are made explicit and correspond to the symmetrisation with respect to particles permutations. The expressions for $\mathcal{V}$ read
    \begin{align}
        \begin{split}
            \mathcal{V}_{\mu_1\nu_1\mu_2\nu_2\mu_3\nu_3} (p_1,p_2,p_3)  = & \cfrac{(-1)\,\kappa^{5-d}}{2^3 } \Bigg[ (p_1)_{\mu_2} (p_2)_{\mu_1} \eta_{\mu_3\nu_1}\eta_{\nu_3\nu_2} \\
            & \hspace{20pt} + \cfrac12\, (p_1)_{\mu_2} (p_3)_{\nu_2} \left(  \eta_{\mu_1\mu_3}\eta_{\nu_1\nu_3} - \cfrac{1}{d - 2} \, \eta_{\mu_1\nu_1}\eta_{\mu_3\nu_3} \right) \\
            & \hspace{20pt} - (p_1 \cdot p_2) \left( \eta_{\nu_1\mu_2}\eta_{\nu_2\mu_3}\eta_{\nu_3\mu_1} -\cfrac{1}{d-2} \, \eta_{\mu_1\nu_1}\,\eta_{\mu_2\nu_3}\eta_{\mu_3\nu_2} \right) + \cdots \Bigg],\\
            \mathcal{V}_{\alpha}{}^{\rho\sigma}{}_{\mu_1\nu_1\mu_2\nu_2} (p_1,p_2) =& - \cfrac{i \, \kappa^{4-d}}{2^2}  \Bigg[ (p_1)_\alpha \delta^\rho_{\mu_1} \delta^\sigma_{\mu_2} \eta_{\nu_1\nu_2} \\
            & \hspace{20pt} - (p_1)^\rho \left( \eta_{\alpha\mu_1} \delta^\sigma_{\mu_2} \eta_{\nu_1\nu_2} - \cfrac{1}{d-2} \, \eta_{\alpha\mu_2} \delta^\sigma_{\nu_2} \eta_{\mu_1\nu_1}  \right)\\
            & \hspace{20pt} + (p_1)_{\mu_2}  \left( \eta_{\alpha\mu_1} \, \delta^\sigma_{\nu_1} \delta^\rho_{\nu_2} - \cfrac{1}{d-2} \, \delta^\sigma_{\alpha} \eta_{\mu_1\nu_1} \delta^\rho_{\nu_2} \right) + \cdots \Bigg] , \\
            \mathcal{V}_{\alpha_1}{}^{\rho_1\sigma_1}{}_{\alpha_2}{}^{\rho_2\sigma_2}{}_{\mu\nu} =& - \cfrac{\kappa^{3-d}}{2^3} \,  \left[ \left(\delta_{\alpha_1}^{\rho_2}\delta_{\alpha_2}^{\rho_1} - \cfrac{1}{d-1} \, \delta_{\alpha_1}^{\rho_1}\delta_{\alpha_2}^{\rho_2} \right) \delta^{\sigma_1}_{\mu}\delta^{\sigma_2}_{\nu} + \cdots \right] .
        \end{split}
    \end{align}
    
\end{theorem}

Let us clarify the choice of functions $\mathcal{V}$ in the theorem above. By the construction, such operators enjoy all the required symmetries. Namely, they are symmetric with respect to permutation of index pairs $\mu_i\leftrightarrow\nu_u$, $\rho_i\leftrightarrow\sigma_i$, and with respect to particles permutations $(\mu_i,\nu_i,p_i)\leftrightarrow(\mu_j,\nu_j,p_j)$ and $(\alpha_i,\rho_i,\sigma_i)\leftrightarrow(\alpha_i,\rho_i,\sigma_i)$. Since we separated the factors associated with particle permutations, one can easily account for them when deriving the Feynman rules. Consequently, pure $\mathcal{V}$ functions will enter the expression for verities.

To use the action \eqref{Cheung-Remmen_action_perturbative} with the path integral, one shall address the gauge fixing. The action is still explicitly gauge invariant. The operator $\mathcal{P}_{\mu\nu\alpha\beta}$ takes the following form in terms of the Nieuwenhuizen operators:
\begin{align}
    \mathcal{P}_{\mu\nu\alpha\beta} &= \kappa^{4-d} \, p^2 \left(\cfrac{d-5}{d-2} \, P^0_{\mu\nu\alpha\beta} + P^2_{\mu\nu\alpha\beta} - \cfrac{d-1}{d-2} \, \overline{P}^0_{\mu\nu\alpha\beta} -\cfrac{1}{d-2} \, \overline{\overline{P}}^0_{\mu\nu\alpha\beta} \right) .
\end{align}
Consequently, the operator $\mathcal{P}_{\mu\nu\alpha\beta}$ is non-invertible. On the contrary, $\mathcal{P}_\text{B}$ is invertible and $\mathcal{P}_\text{B}^{-1}$ takes the following form:
\begin{align}
    \left(\mathcal{P}_\text{B}^{-1}\right)^{\lambda_1}{}_{\mu_1\nu_1}{}^{\lambda_2}{}_{\mu_2\nu_2} = \kappa^{d-2} \, \cfrac12 \, \mathcal{O}^{\lambda_1}{}_{\mu_1\nu_1}{}^{\lambda_2}{}_{\mu_2\nu_2} (\eta).
\end{align}
Consequently, one still needs to introduce a gauge fixing term to construct a propagator for $\goh$.

We fix the gauge with the BRST formalism. Since we discussed the formalism in the previous section, we present a much shorter discussion here, focusing only on the gauge fixing with a choice of background geometry. We already obtained all the required formulae to introduce the BRST transformations. We shall only present them in a form more suitable for the theory under discussion. Namely, the following formulae hold for the Cheung-Remmen metric:
\begin{align}
    \delta\gog^{\mu\nu} & = \mathcal{L}_\zeta \left( \sqrt{-g} \, g^{\mu\nu} \right) = - \left[ \gog^{\mu\lambda} \pd_\lambda \zeta^\nu + \gog^{\nu\lambda} \pd_\lambda \zeta^\mu - \pd_\lambda\left( \zeta^\lambda \gog^{\mu\nu} \right) \right] .
\end{align}

The BRST transformation of the auxiliary field is more subtle. Gauge models, including gravity, are gauge invariant only on the mass shell. Consequently, our model is gauge invariant and BRST invariant only on the mass shell. At the same time, the auxiliary field $A$ does not exist on the mass shell since it reduces to the Cheung-Remmen metric:
\begin{align}
    A^\lambda_{\alpha\beta} &= \mathcal{O}^\lambda{}_{\alpha\beta}{}^{\lambda_1\alpha_1\beta_1}(\gog)  \pd_{\lambda_1} \gog_{\alpha_1\beta_1} .
\end{align}
Because of this, one can calculate the action of the BRST operator $\delta_B$ on the auxiliary field $A$ only on the mass shell, but the result can hardly provide any additional information about the structure of the model.

For the sake of completeness, let us show how $\delta_B$ acts on the auxiliary field when it is on the mass shell. One can express the auxiliary field in terms of the Christoffel symbols calculated for the original metric:
\begin{align}
    \begin{split}
        A^\lambda_{\alpha\beta} &= \mathcal{O}^\lambda{}_{\alpha\beta}{}^{\lambda_1\alpha_1\beta_1}(\gog)  \pd_{\lambda_1} \gog_{\alpha_1\beta_1} = \Gamma^\lambda_{\alpha\beta}[g] - \cfrac12\, \left( \delta^\lambda_\alpha \, \Gamma^\sigma_{\beta\sigma}[g] + \delta^\lambda_\beta \, \Gamma^\sigma_{\alpha\sigma}[g] \right) , \\
        \Gamma^\lambda_{\alpha\beta}[g] & =  A^\lambda_{\alpha\beta} - \cfrac{1}{d-1}\left( \delta^\lambda_\alpha A^\sigma_{\beta\sigma} + \delta^\lambda_\beta A^\sigma_{\alpha\sigma} \right) .
    \end{split}
\end{align}
One can use these relations to calculate the variation (with respect to the gauge transformation) of $A^\lambda_{\alpha\beta}$:
\begin{align}
    \delta A^\lambda_{\mu\nu} & = \cfrac{1}{2}\left[ \nabla_\mu\nabla_\nu\zeta^\lambda + \nabla_\nu\nabla_\nu\zeta^\lambda - \delta^\lambda_\mu \nabla_\nu\nabla_\sigma\zeta^\sigma - \delta^\lambda_\nu \nabla_\mu\nabla_\sigma\zeta^\sigma + R_\mu{}{}^\lambda{}_{\nu\sigma}\zeta^\sigma + R_\nu{}^\lambda{}_{\mu\sigma}\zeta^\sigma\right] .
\end{align}
Consequently, the action of the BRST operator on the auxiliary field on the mass shell reads:
\begin{align}
    \delta_B A^\lambda_{\mu\nu} & = \cfrac{1}{2}\left[ \nabla_\mu\nabla_\nu c^\lambda + \nabla_\nu\nabla_\nu c^\lambda - \delta^\lambda_\mu \nabla_\nu\nabla_\sigma c^\sigma - \delta^\lambda_\nu \nabla_\mu\nabla_\sigma c^\sigma + R_\mu{}{}^\lambda{}_{\nu\sigma} c^\sigma + R_\nu{}^\lambda{}_{\mu\sigma} c^\sigma\right] .
\end{align}

The following theorem gives the BRST transformation of the original Cheung-Remmen variables
\begin{theorem}

    {~}

    \noindent The BRST transformations for variables $\gog$ and $A$ reads:
    \begin{align}
        \begin{split}
            \delta_B \gog^{\mu\nu} & = - \left[ \gog^{\mu\lambda} \pd_\lambda c^\nu + \gog^{\nu\lambda}\pd_\lambda c^\mu -\pd_\lambda\left( \gog^{\mu\nu} c^\lambda \right) \right] , \\
            \delta_B c^\mu &= - c^\sigma \pd_\sigma c^\mu\, ,\\
            \delta_B \overline{c}_\mu &= b_\mu \, ,\\
            \delta b_\mu &= 0 \,. 
        \end{split}
    \end{align}
    
\end{theorem}

\noindent In full analogy with the previous section, field $\overline{c}_\mu$ and $c^\mu$ have the canonical mass dimension. We denote the auxiliary field required by the BRST formalism as $b_\mu$ in order to avoid confusion with the auxiliary field $B^{\lambda}_{\mu\nu}$. The field also has a non-canonical mass dimension $3$, and we will integrate it out in the functional integral.

The transition to the perturbative variables is a linear field redefinition.

\begin{theorem} 

    {~} 
    
    \noindent The BRST transformation for perturbative variables reads:
    \begin{align}
        \begin{split}
            \delta_B \goh^{\mu\nu} & = \goh^{\mu\lambda} \pd_\lambda c^\nu + \goh^{\mu\lambda} \pd_\lambda c^\nu - \pd_\lambda\left( \goh^{\mu\nu} c^\lambda\right)  -\cfrac{1}{\kappa}\left[ \eta^{\mu\lambda} \pd_\lambda c^\nu + \eta^{\mu\lambda} \pd_\lambda c^\nu - \eta^{\mu\nu} \pd_\lambda c^\lambda \right] \, , \\
            \delta_B c^\mu &= - c^\sigma \pd_\sigma c^\mu\, ,\\
            \delta_B \overline{c}_\mu &= b_\mu \, ,\\
            \delta b_\mu &= 0 \,. 
        \end{split}
    \end{align}
\end{theorem}

Further, we introduce the gauge fixing term of the following form:
\begin{align}
    \mathcal{A}_\text{gf} = \kappa^{4-d} \delta_B \int d^d x ~ \overline{c}_\mu Y^{\mu\nu} \left( \mathcal{G}_\nu - \cfrac{\varepsilon}{2} \, \kappa^2 \, b_\nu \right) .
\end{align}
The additional factor $\kappa^{4-d}$ appears because we use a different spacetime dimension. We are already considering the flat background to remove an undesirable linear coupling of the auxiliary field $A$ to an external source, so we make the background BRST invariant and assume that $Y^{\mu\nu}$ can depend on the background only as well. Therefore, the gauge fixing term takes the following form:
\begin{align}
    \mathcal{A}_\text{gf} = \kappa^{4-d} \int d^d x \Big[ -\cfrac{\varepsilon}{2}\,\kappa^2 \, b_\mu Y^{\mu\nu} b_\nu + b_\mu Y^{\mu\nu} \mathcal{G}_\nu - \overline{c}_\mu Y^{\mu\nu} \delta_B \mathcal{G}_\nu \Big] .
\end{align}
The corresponding generating functional is a Gaussian integral with respect to the auxiliary field $b_\mu$, so that one can integrate this field out.

\begin{theorem}

    {~}

    \noindent The generating functional for the Cheung-Remmen action with the given gauge fixing term takes the following form:
    \begin{align}
        \begin{split}
            \mathcal{Z} &= \int \mathcal{D}[\goh] \mathcal{D}[B] \mathcal{D}[\overline{c}] \mathcal{D}[c] \mathcal{D}[b] \exp\left[ i\, \mathcal{A}_\text{CR}[\goh,B] + i \, \mathcal{A}_\text{gf}[\goh,\overline{c},c,b]\right]\\
            & = \int \mathcal{D}[\goh] \mathcal{D}[B] \mathcal{D}[\overline{c}] \mathcal{D}[c] \exp\left[ i\, \mathcal{A}_\text{CR}[\goh,B] + i \kappa^{4-d} \int d^d x \left[ \cfrac{1}{2\varepsilon\kappa^2} \,\mathcal{G}_\mu Y^{\mu\nu} \mathcal{G}_\nu - \overline{c}_\mu Y^{\mu\nu} \delta_B \mathcal{G}_\nu \right] \right] \,.
        \end{split}
    \end{align}
    
\end{theorem}

For simplicity we set $Y^{\mu\nu}$ and $\mathcal{G}^\mu$ to be the following:
\begin{align}
    Y^{\mu\nu} & = \eta^{\mu\nu} , & \mathcal{G}^\mu & = \kappa\, \pd_\nu \goh^{\mu\nu}  .
\end{align}
Consequently, the BRST action on $\mathcal{G}^\mu$ reduces to a simple form:
\begin{align}
    \delta_B \mathcal{G}^\mu = - \square c^\mu + \kappa \, \pd_\nu\left[ \goh^{\mu\lambda}\pd_\lambda c^\nu + \goh^{\nu\lambda}\pd_\lambda c^\mu - \pd_\lambda( \goh^{\mu\nu} c^\lambda) \right].
\end{align}
The remaining calculations are straightforward, and the following theorem specifies their result.

\begin{theorem}

    {~}

    \noindent The generating functional for the Cheung-Remmen action with the gauge fixing term
    \begin{align}
        \mathcal{A}_\text{gf} &= \kappa^{4-d} \delta_B \int d^d x ~ \overline{c}_\mu \left( \pd_\nu \goh^{\mu\nu}  - \cfrac{\varepsilon}{2} \, \kappa^2 \, b^\mu \right) 
    \end{align}
    takes the following form
    \begin{align}
        \begin{split}
            \mathcal{Z} &= \int\mathcal{D}[\goh,B,b,\overline{c},c] \exp\Big[ i \mathcal{A}_\text{CR}[\goh,B] + i \, \mathcal{A}_\text{gf}[\goh,b,\overline{c},c] \Big] \\
            & = \int\mathcal{D}[\goh,B,\overline{c},c] \exp\Big[ i \, \mathcal{A}_{\goh\goh+\text{gf}} + i\, \mathcal{A}_{BB} + i\, \mathcal{A}_{\goh\goh\goh} + i\, \mathcal{A}_{B\goh\goh} + i \,\mathcal{A}_{BB\goh} + i \, \mathcal{A}_\text{gh} + i \, \mathcal{A}_{\text{gh}-\goh} \Big] .
        \end{split}
    \end{align}
    The actions $\mathcal{A}_{BB}$,$\mathcal{A}_{\goh\goh\goh}$,$\mathcal{A}_{B\goh\goh}$, and $\mathcal{A}_{BB\goh}$ were specified above. $\mathcal{A}_{\goh\goh+\text{gf}}$ is the action describing $\goh$ propagation with the gauge fixing term:
    \begin{align}
        \begin{split}
            \mathcal{A}_{\goh\goh+\text{gf}} & = \mathcal{A}_{\goh\goh} + \cfrac12\, \cfrac{\kappa^{4-d}}{\varepsilon} \int \cfrac{d^dp}{(2\pi)^d} \, \goh^{\mu_1\nu_1}(-p) \goh^{\mu_2\nu_2}(p) \, \cfrac14 \left[ \eta_{\mu_1\mu_2} p_{\nu_1} p_{\nu_2} + \cdots \right] \\
            & = \int \cfrac{d^d p}{(2\pi)^d} ~ \goh^{\mu_1\nu_1}(-p) \, \goh^{\mu_2\nu_2}(p) \, \cfrac12\, (\mathcal{P}_\varepsilon)_{\mu_1\nu_1\mu_2\nu_2} \,.
        \end{split}
    \end{align}
    The operator $(\mathcal{P}_\varepsilon)$ takes the following form in terms of the Nieuwenhuizen operators:
    \begin{align}
        (\mathcal{P}_\varepsilon)_{\mu\nu\alpha\beta} &= \kappa^{4-d} \, p^2 \left(\cfrac{d-5}{d-2} \, P^0_{\mu\nu\alpha\beta} - \cfrac{1}{2\,\varepsilon} \, P^1_{\mu\nu\alpha\beta} + P^2_{\mu\nu\alpha\beta} - \left(\cfrac{d-1}{d-2} + \cfrac{1}{\varepsilon} \right) \overline{P}^0_{\mu\nu\alpha\beta} -\cfrac{1}{d-2} \, \overline{\overline{P}}^0_{\mu\nu\alpha\beta} \right) .
    \end{align}
    It is invertible, and the inverse operator reads:
    \begin{align}
        (\mathcal{P}_\varepsilon^{-1})_{\mu\nu\alpha\beta} &= \kappa^{d-4} \,\cfrac{1}{p^2} \left( \left(- 2 - 3 \varepsilon \right) \, P^0_{\mu\nu\alpha\beta} - 2\,\varepsilon \, P^1_{\mu\nu\alpha\beta} + P^2_{\mu\nu\alpha\beta} - \varepsilon \overline{P}^0_{\mu\nu\alpha\beta} - \varepsilon \, \overline{\overline{P}}^0_{\mu\nu\alpha\beta} \right) .
    \end{align}

    \noindent $\mathcal{A}_\text{gh}$ is the ghost action:
    \begin{align}
        \mathcal{A}_\text{gh} = \kappa^{4-d} \int d^d x \, \overline{c}_\mu \square c^\mu  = \kappa^{4-d} \int \cfrac{d^d p}{(2\pi)^d} \, \overline{c}_\nu(-p) c_\mu(p) \left[ - p^2 \, \eta_{\mu\nu}\right] \,.
    \end{align}
    
    \noindent $\mathcal{A}_{\text{gh}-\goh}$ is the action describing a coupling between ghosts and perturbations $\goh$:
    \begin{align}
        \begin{split}
            \mathcal{A}_{\text{gh}-\goh} &=  \kappa^{5-d} \int d^d x \, \goh^{\mu\nu} \left[ \pd_\lambda \overline{c}_\mu \pd_\nu c^\lambda + \pd_\mu \overline{c}_\lambda \pd_\nu c^\lambda + \pd_\lambda \pd_\mu \overline{c}_\nu c^\lambda \right]\\
            & = \int \cfrac{d^d k}{(\pi)^d}\,\cfrac{d^d p_1}{(\pi)^d}\,\cfrac{d^d p_2}{(\pi)^d} \, (2\pi)^d \delta(k+p_1+p_2) \, \goh^{\mu_1\nu_1}(k) \, \overline{c}_\beta(p_2) \, c^\alpha(p_1)  \\
            & \hspace{20pt} \times \kappa^{5-d} \left[ (p_2)_\alpha (p_1)_\nu \delta^\beta_\mu + (p_2)_\mu (p_1)_\nu \delta_\alpha^\beta + (p_2)_\alpha (p_2)_\mu \delta_\nu^\beta \right].
        \end{split}
    \end{align}
    
\end{theorem}

The theorem is the main result of this section. One can use these formulae to obtain a finite set of Feynman rules for general relativity in Cheung-Remmen variables in arbitrary dimensions. For the sake of illustration, we present a set of Feynman rules in $d=4$ with the gauge fixing parameter $\varepsilon=-1/2$, in which case the interaction rules take a remarkably simple form.

The $\goh$ perturbations propagator reads:
\begin{align}\label{CR_propagator_hh}
    \begin{gathered}
        \begin{fmffile}{CR_hh}
            \begin{fmfgraph*}(30,30)
                \fmfleft{L}
                \fmfright{R}
                \fmf{dbl_wiggly}{L,R}
                \fmflabel{$\mu\nu$}{L}
                \fmflabel{$\alpha\beta$}{R}
            \end{fmfgraph*}
        \end{fmffile}
    \end{gathered}
    \hspace{20pt} = \cfrac{i}{p^2} \, \kappa^{4-d} \, \cfrac12\, \left[ \eta_{\mu\alpha} \eta_{\nu\beta} + \eta_{\mu\beta} \eta_{\nu\alpha} - \eta_{\mu\nu} \eta_{\alpha\beta} \right] .
\end{align}
The auxiliary field $B$ propagator reads:
\begin{align}
    \begin{gathered}
        \begin{fmffile}{CR_BB}
            \begin{fmfgraph*}(30,30)
                \fmfleft{L}
                \fmfright{R}
                \fmf{dbl_dashes}{L,R}
                \fmflabel{${}^{\lambda_1}_{\alpha_1\beta_1}$}{L}
                \fmflabel{${}^{\lambda_2}_{\alpha_2\beta_2}$}{R}
            \end{fmfgraph*}
        \end{fmffile}
    \end{gathered}
    \hspace{25pt} = i\, \cfrac{\kappa^{d-2}}{8}\left[ \left( \delta^{\lambda_1}_{\alpha_2} \delta^{\lambda_2}_{\alpha_1}\eta_{\beta_1\beta_2} + \cdots \right) - \eta^{\lambda_1\lambda_2} \left( \eta_{\alpha_1\alpha_2}\eta_{\beta_1\beta_2} + \cdots \right) + \eta^{\lambda_1\lambda_2}\,\eta_{\alpha_1\beta_1}\,\eta_{\alpha_2\beta_2} \right] .
\end{align}
The Faddeev-Popov ghost field propagator reads:
\begin{align}
    \begin{gathered}
        \begin{fmffile}{CR_cc}
            \begin{fmfgraph*}(30,30)
                \fmfleft{L}
                \fmfright{R}
                \fmf{dots_arrow}{L,R}
                \fmflabel{$\mu$}{L}
                \fmflabel{$\nu$}{R}
            \end{fmfgraph*}
        \end{fmffile}
    \end{gathered}
    \hspace{20pt} = -\cfrac{i}{p^2} \, \eta_{\mu\nu} \, .
\end{align}
Ghost-perturbations vertex reads:
\begin{align}
    \nonumber \\
    \begin{gathered}
        \begin{fmffile}{CR_cch}
            \begin{fmfgraph*}(40,40)
                \fmfleft{L}
                \fmfright{R1,R2}
                \fmf{dbl_wiggly}{L,V}
                \fmf{dots_arrow}{R1,V,R2}
                \fmfdot{V}
                \fmflabel{$\mu\nu$}{L}
                \fmflabel{$\alpha,p_1$}{R1}
                \fmflabel{$\beta,p_2$}{R2}
            \end{fmfgraph*}
        \end{fmffile}
    \end{gathered}
    \hspace{30pt} = i\, \kappa^{5-d} \left[ (p_2)_\alpha (p_1)_\nu \delta^\beta_\mu + (p_2)_\mu (p_1)_\nu \delta_\alpha^\beta + (p_2)_\alpha (p_2)_\mu \delta_\nu^\beta \right] .
\end{align}
Cubic $\goh^3$ vertex reads:
\begin{align}
    \nonumber \\
    \hspace{20pt}
    \begin{gathered}
        \begin{fmffile}{CR_hhh}
            \begin{fmfgraph*}(40,40)
                \fmfleft{L}
                \fmfright{R1,R2}
                \fmf{dbl_wiggly}{L,V}
                \fmf{dbl_wiggly}{R1,V}
                \fmf{dbl_wiggly}{R2,V}
                \fmfdot{V}
                \fmflabel{$\mu_1\nu_1,p_1$}{L}
                \fmflabel{$\mu_2\nu_2,p_2$}{R1}
                \fmflabel{$\mu_3\nu_3,p_3$}{R2}
            \end{fmfgraph*}
        \end{fmffile}
    \end{gathered}
    \hspace{20pt} =& \cfrac{-i \, \kappa^{5-d} }{2^3 } \Bigg[ (p_1)_{\mu_2} (p_2)_{\mu_1} \eta_{\mu_3\nu_1}\eta_{\nu_3\nu_2} + \cfrac12\, (p_1)_{\mu_2} (p_3)_{\nu_2} \left(  \eta_{\mu_1\mu_3}\eta_{\nu_1\nu_3} - \cfrac{1}{d - 2} \, \eta_{\mu_1\nu_1}\eta_{\mu_3\nu_3} \right) \\
    & - (p_1 \cdot p_2) \left( \eta_{\nu_1\mu_2}\eta_{\nu_2\mu_3}\eta_{\nu_3\mu_1} -\cfrac{1}{d-2} \, \eta_{\mu_1\nu_1}\,\eta_{\mu_2\nu_3}\eta_{\mu_3\nu_2} \right) + \cdots \Bigg] . \nonumber
\end{align}
Cubic $B \goh^2$ vertex reads:
\begin{align}
    \nonumber \\
    \begin{gathered}
        \begin{fmffile}{CR_Bhh}
            \begin{fmfgraph*}(40,40)
                \fmfleft{L}
                \fmfright{R1,R2}
                \fmf{dbl_dashes}{L,V}
                \fmf{dbl_wiggly}{R1,V}
                \fmf{dbl_wiggly}{R2,V}
                \fmfdot{V}
                \fmflabel{${}_{\alpha}^{\rho\sigma}$}{L}
                \fmflabel{$\mu_1\nu_1,p_1$}{R1}
                \fmflabel{$\mu_2\nu_2,p_2$}{R2}
            \end{fmfgraph*}
        \end{fmffile}
    \end{gathered}
    \hspace{20pt}= & \cfrac{1}{2^2} \, \kappa^{4-d}  \Bigg[ (p_1)_\alpha \eta_{\rho\mu_1} \eta_{\sigma\mu_2} \eta_{\nu_1\nu_2}  - (p_1)_\rho \left( \eta_{\alpha\mu_1} \eta_{\sigma\mu_2} \eta_{\nu_1\nu_2} - \cfrac{1}{d-2} \, \eta_{\alpha\mu_2} \eta_{\sigma\nu_2} \eta_{\mu_1\nu_1}  \right)\\
    &  + (p_1)_{\mu_2}  \left( \eta_{\alpha\mu_1} \, \eta_{\sigma\nu_1} \eta_{\rho\nu_2} - \cfrac{1}{d-2} \, \eta_{\alpha\sigma} \eta_{\mu_1\nu_1} \eta_{\nu_2 \rho} \right) + \cdots \Bigg] . \nonumber
\end{align}
Cubic $B^2 \goh$ vertex reads:
\begin{align}
    \nonumber \\
    \begin{gathered}
        \begin{fmffile}{CR_BBh}
            \begin{fmfgraph*}(40,40)
                \fmfleft{L}
                \fmfright{R1,R2}
                \fmf{dbl_wiggly}{L,V}
                \fmf{dbl_dashes}{R1,V}
                \fmf{dbl_dashes}{R2,V}
                \fmfdot{V}
                \fmflabel{${}_{\alpha_1}^{\rho_1\sigma_1}$}{R1}
                \fmflabel{${}_{\alpha_2}^{\rho_2\sigma_2}$}{R2}
                \fmflabel{$\mu\nu$}{L}
            \end{fmfgraph*}
        \end{fmffile}
    \end{gathered}
    \hspace{20pt} =& - \cfrac{1}{2^3} \, \kappa^{3-d} \,  \left[ \left(\delta_{\alpha_1}^{\rho_2}\delta_{\alpha_2}^{\rho_1} - \cfrac{1}{d-1} \, \delta_{\alpha_1}^{\rho_1}\delta_{\alpha_2}^{\rho_2} \right) \delta^{\sigma_1}_{\mu}\delta^{\sigma_2}_{\nu} + \cdots \right].
\end{align}

Concluding this section, let us briefly comment on the limitations of the Cheung-Remmen variables. There are three main disadvantages associated with the Cheung-Remmen variables that are impossible to remove without a significant modification of the formalism. First and foremost, the variables $\sqrt{-\gog} \, \gog^{\mu\nu}$ only bring the Hilbert action to the polynomial form, but fail to do so for higher derivative gravity. Secondly, it is crucial to perform the perturbative expansion of the Cheung-Remmen variables about the flat background, since any other background introduces a linear coupling between an external source (background metric derivatives) and the auxiliary variables $B^{\lambda}_{\mu\nu}$. This condition further limits applicability because the corresponding perturbative expansion applies only to models that admit Minkowski space as a solution.

Last but not least, the coupling between the Cheung-Remmen variables and matter is obscured. If one works with the conventional metric $g_{\mu\nu}$ then typically couples to a factor $\sqrt{-g} \, g^{\mu_1\nu_1} \cdots g^{\mu_n\nu_n}$. Here $n$ depends on the type of matter, for instance, $n=1$ for scalars and $n=2$ for vectors. The same factor calculated in Cheung-Remmen variables receives an additional power of the metric determinant:
\begin{align}
    \sqrt{-g} \, g^{\mu_1\nu_1} \cdots g^{\mu_n\nu_n} = \left( \sqrt{-\gog} \right)^{\frac{2(n-1)}{d-2}} \, \gog^{\mu_1\nu_1} \cdots \gog^{\mu_n\nu_n}.
\end{align}
The previous works \cite{Latosh:2022ydd,Latosh:2023zsi,Latosh:2024lhl} found algorithms to calculate and arbitrary term perturbative expansion of $\sqrt{-g} g^{\mu_1\nu_1} \cdots g^{\mu_n\nu_n}$. One can use the same algorithms to expand $ \sqrt{-\gog} \, \gog^{\mu_1\nu_1} \cdots \gog^{\mu_n\nu_n}$. However, there are no known effective algorithms for working with $\left( \sqrt{-\gog} \right)^{\frac{2(n-1)}{d-2}}$. Certainly, one can use the standard Taylor expansion to calculate the corresponding couplings, but this approach appears to be highly computationally expensive. Further analysis of the existing works which used Cheung-Remmen variables, such as but not limited to \cite{Cheung:2020gyp,Cheung:2020sdj,Abreu:2020lyk,Britto:2021pud}, may provide further leverage in deriving a more efficient way to deal with the corresponding couplings.

For these reasons, the new version of \texttt{FeynGrav}, discussed in the following section, only implements the Cheung-Remmen variables and does not address its coupling to matter. The research for a generalisation of the Cheung-Remmen variables more optimised for the perturbative treatment lies far beyond the scope of this paper.

\section{New features of FeynGrav}\label{Section_FG}

The new version of \texttt{FeynGrav} includes the following changes. Firstly, it implements the new version of Feynman rules with a finite number of interaction vertices in the ghost sector for general relativity and quadratic gravity. Secondly, we implement the Feynman rules for quadratic gravity with the higher derivative gauge fixing term and a finite number of interaction vertices. Thirdly, we implement Feynman rules for Cheung-Remmen variables. Lastly, we introduce some minor changes aimed at improving the package usability. Namely, we introduce commands to operate with the Nieuwenhuizen operators more efficiently, and modify descriptions of some existing commands.

Firstly, let us discuss the implementation of the interaction rules with a finite number of vertices in the ghost sector. The previous version of \texttt{FeynGrav} implemented a set of rules with an infinite number of ghost-graviton vertices. The gauge fixing parameter is also entered both in the graviton propagator and the vertices. The present version introduces the gauge fixing parameter only in the propagators, while the expressions for the vertices are free of it.

We use the same global variable \texttt{GaugeFixingEpsilon} for the gauge fixing parameter. At the initialisation of the package, it is set to $2$, so the user can work with the simple graviton propagator \eqref{GR_propagator_simple}. One can reset the variable with the command
\begin{align}
    \texttt{GaugeFixingEpsilon=.} 
\end{align}
and give it any other value afterwards. The syntax of the command \texttt{GravitonVertex}, which returns graviton vertices, remained unaffected, but the involved expressions no longer contain the gauge fixing parameter.

The new version of \texttt{FeynGrav} provides a command for the Faddeev-Popov vector ghost propagator:
\begin{align}
    \begin{gathered}
        \begin{fmffile}{Ghost_Propagator}
            \begin{fmfgraph*}(30,30)
                \fmfleft{L}
                \fmfright{R}
                \fmf{dots_arrow}{L,R}
                \fmflabel{$\mu$}{L}
                \fmflabel{$\nu$}{R}
            \end{fmfgraph*}
        \end{fmffile}
    \end{gathered}
    \hspace{20pt} = \texttt{GhostVectorPropagator}[\mu,\nu,p].
\end{align}
We implemented this command for usability. One could use a modified command \texttt{GhostPropagator} from \texttt{FeynCalc}, but this appears to be highly inconvenient. Moreover, we introduce a similar command for the higher derivative gauge fixing term.

The set of Feynman rules contains a single graviton-ghost vertex, so there is no need for multiple libraries describing such a coupling. The command is built in the main package and corresponds to the following vertex:
\begin{align}
\nonumber \\
    \begin{gathered}
        \begin{fmffile}{Ghost_Graviton_Vertex_1}
            \begin{fmfgraph*}(40,40)
                \fmfleft{L}
                \fmfright{R1,R2}
                \fmf{dbl_wiggly}{L,V}
                \fmf{dots_arrow}{R1,V}
                \fmf{dots_arrow}{V,R2}
                \fmfdot{V}
                \fmflabel{$\rho\sigma,k$}{L}
                \fmflabel{$\mu,p_1$}{R1}
                \fmflabel{$\nu,p_2$}{R2}
            \end{fmfgraph*}
        \end{fmffile}
    \end{gathered} 
    = \texttt{GravitonGhostVertex}[\rho,\sigma,k,\mu,p_1,\nu,p_2] . \\ \nonumber
\end{align}

For the quadratic gravity, we implemented similar changes. The \texttt{QuadraticGravityPropagator} command inherits its syntax from the previous versions, but does not use the specified gauge fixing term with the same global gauge fixing parameter \texttt{GaugeFixingEpsilon}. The command \texttt{QuadraticGravityPropagator} gives the corresponding vertices described in the previous publications \cite{Latosh:2022ydd,Latosh:2023zsi,Latosh:2024lhl}, but it is not free from the gauge fixing parameter.

The higher derivative gauge fixing term requires the introduction of new commands. First and foremost, we introduced four gauge fixing parameters for the gauge fixing term: $\varepsilon$, $\varepsilon_0$, $\varepsilon_1$, and $\varepsilon_2$. For the case of the flat background, parameters $\varepsilon_1$ and $\varepsilon_2$ enter all the expressions as a sum. Therefore, for the sake of simplicity, we set $\varepsilon_2=0$. Because of this, \texttt{FeynGrav} has three higher derivative gauge fixing parameters: \texttt{GaugeFixingEpsilonHD}, \texttt{GaugeFixingEpsilonHD0}, and \texttt{GaugeFixingEpsilonHD1} corresponding to $\varepsilon$, $\varepsilon_0$, and $\varepsilon_1$. In contrast to the conventional gauge fixing, we do not define these parameters at the initialisation, and they remain unspecified global constants. Commands defining gravity propagator, ghost propagator, and graviton-ghost vertex are: \texttt{QuadraticGravityPropagatorHD}, \texttt{GhostVectorPropagatorHD}, and \texttt{GravitonGhostVertexHD}. They inherit the exact syntax as their conventional counterparts.

Further, we introduce commands to work with the Cheung-Remmen variables. There are two propagators, one for the perturbative variable $\goh^{\mu\nu}$, and one for the auxiliary variable $B^\lambda_{\alpha\beta}$:
\begin{align}
    &
    \begin{gathered}
        \begin{fmffile}{CR_hh}
            \begin{fmfgraph*}(30,30)
                \fmfleft{L}
                \fmfright{R}
                \fmf{dbl_wiggly}{L,R}
                \fmflabel{$\mu\nu$}{L}
                \fmflabel{$\alpha\beta$}{R}
            \end{fmfgraph*}
        \end{fmffile}
    \end{gathered}
    \hspace{20pt} = \texttt{GravitonPropagatorCR} [\mu,\nu,\alpha,\beta,p],\\
    &
    \begin{gathered}
        \begin{fmffile}{CR_BB}
            \begin{fmfgraph*}(30,30)
                \fmfleft{L}
                \fmfright{R}
                \fmf{dbl_dashes}{L,R}
                \fmflabel{${}^{\lambda_1}_{\alpha_1\beta_1}$}{L}
                \fmflabel{${}^{\lambda_2}_{\alpha_2\beta_2}$}{R}
            \end{fmfgraph*}
        \end{fmffile}
    \end{gathered}
    \hspace{20pt} = \texttt{GravitonPropagatorAuxiliaryCR}[\lambda_1,\alpha_1\beta_1,\lambda_2,\alpha_2\beta_2],
    \end{align}
The gauge fixing term for the Cheung-Remmen variables only influences the propagator for $\goh$ and does not enter any other expressions. This parameter, called \texttt{GaugeFixingEpsilonCR}, is a global constant. Similar to the gauge fixing parameter of general relativity, it is set to be $-1/2$ at the initialisation, so the $\goh$ propagator has a simple form \eqref{CR_propagator_hh}. The only difference present in the particular realisation of these rules is the difference in the overall factor in the $\goh$ propagator. The propagator contains a factor $\kappa^{d-4}$ which we omit for the sake of brevity. 

The propagator for $B^\lambda_{\alpha\beta}$ is algebraic (since this is an auxiliary field), so it does not contain any momenta. The model also includes the ghost propagator, which matches the ghost propagator in general relativity and is given by the same command.

The following rules describe interaction vertices:
    \begin{align}
    &
    \begin{gathered}
        \begin{fmffile}{CR_hhh}
            \begin{fmfgraph*}(35,35)
                \fmfleft{L}
                \fmfright{R1,R2}
                \fmf{dbl_wiggly}{L,V}
                \fmf{dbl_wiggly}{R1,V}
                \fmf{dbl_wiggly}{R2,V}
                \fmfdot{V}
                \fmflabel{$\mu_1\nu_1,p_1$}{L}
                \fmflabel{$\mu_2\nu_2,p_2$}{R1}
                \fmflabel{$\mu_3\nu_3,p_3$}{R2}
            \end{fmfgraph*}
        \end{fmffile}
    \end{gathered}
    = \texttt{GravitonVertexCRhhh}[\mu_1,\nu_1,p_1,\mu_2,\nu_2,p_2,\mu_3,\nu_3,p_3] , \\ \nonumber \\ \nonumber \\
    &
    \begin{gathered}
        \begin{fmffile}{CR_Bhh}
            \begin{fmfgraph*}(35,35)
                \fmfleft{L}
                \fmfright{R1,R2}
                \fmf{dbl_dashes}{L,V}
                \fmf{dbl_wiggly}{R1,V}
                \fmf{dbl_wiggly}{R2,V}
                \fmfdot{V}
                \fmflabel{${}_{\alpha}^{\rho\sigma}$}{L}
                \fmflabel{$\mu_1\nu_1,p_1$}{R1}
                \fmflabel{$\mu_2\nu_2,p_2$}{R2}
            \end{fmfgraph*}
        \end{fmffile}
    \end{gathered}
    = \texttt{GravitonVertexCRBhh}[\alpha,\rho,\sigma,\mu_1,\nu_1,p_1,\mu_2,\nu_2,p_2] , \\ \nonumber \\ \nonumber \\
    &
    \begin{gathered}
        \begin{fmffile}{CR_BBh}
            \begin{fmfgraph*}(35,35)
                \fmfleft{L}
                \fmfright{R1,R2}
                \fmf{dbl_wiggly}{L,V}
                \fmf{dbl_dashes}{R1,V}
                \fmf{dbl_dashes}{R2,V}
                \fmfdot{V}
                \fmflabel{${}_{\alpha_1}^{\rho_1\sigma_1}$}{R1}
                \fmflabel{${}_{\alpha_2}^{\rho_2\sigma_2}$}{R2}
                \fmflabel{$\mu\nu$}{L}
            \end{fmfgraph*}
        \end{fmffile}
    \end{gathered}
    = \texttt{GravitonVertexCRBBh}[\alpha_1,\rho_1,\sigma_1,\alpha_2,\rho_2,\sigma_2,\mu,\nu]. \\ \nonumber
\end{align}

The new version of \texttt{FeynGrav} received three more commands that help to operate with the Nieuwenhuizen operators (see Appendix \ref{Appendix_Nieuwenhuizen}). Command \texttt{NieuwenhuizenOperator} returns a linear combination with the given weights:
\begin{align}
    \begin{split}
        & \texttt{NieuwenhuizenOperator}[z_1,z_2,z_0,\overline{z}_0,\overline{\overline{z}}_0,\mu,\nu,\alpha,\beta,p] \\
        & = z_1 P^1_{\mu\nu\alpha\beta}(p) + z_2 P^2_{\mu\nu\alpha\beta}(p) + z_0 P^0_{\mu\nu\alpha\beta}(p) + \overline{z}_0 \overline{P}^0_{\mu\nu\alpha\beta}(p) + \overline{\overline{z}}_0 \overline{\overline{P}}^0_{\mu\nu\alpha\beta}(p).
    \end{split}
\end{align}
Command \texttt{NieuwenhuizenOperatorInverse} returns a linear combination of Nieuwenhuizen operators that is inverse to a given operator
\begin{align}
    \begin{split}
        & \texttt{NieuwenhuizenOperatorInverse}[z_1,z_2,z_0,\overline{z}_0,\overline{\overline{z}}_0,\mu,\nu,\alpha,\beta,p] \\
        & = \left( z_1 P^1_{\mu\nu\alpha\beta}(p) + z_2 P^2_{\mu\nu\alpha\beta}(p) + z_0 P^0_{\mu\nu\alpha\beta}(p) + \overline{z}_0 \overline{P}^0_{\mu\nu\alpha\beta}(p) + \overline{\overline{z}}_0 \overline{\overline{P}}^0_{\mu\nu\alpha\beta}(p) \right)^{-1}.
    \end{split}
\end{align}
Command \texttt{NieuwenhuizenOperatorExpansion} expands a given operator in Nieuwenhuizen operators and returns the corresponding coordinates:
\begin{align}
    \begin{split}
        & \texttt{NieuwenhuizenOperatorExpansion}[ \texttt{NieuwenhuizenOperator}[z_1,z_2,z_0,\overline{z}_0,\overline{\overline{z}}_0,\mu,\nu,\alpha,\beta,p] ,\mu,\nu,\alpha,\beta,p] \\
        & = \{z_1,z_2,z_0,\overline{z}_0,\overline{\overline{z}}_0\} .
    \end{split}
\end{align}

Lastly, all \texttt{FeynGrav} commands received updated descriptions. Within \texttt{Mathematica}, a user can call for information about a command. For instance, ``\texttt{?GravitonPropagator}'' displays brief information about the command \texttt{GravitonPropagator}. All information about all commands present in \texttt{FeynGrav} is updated.

\section{Conclusions}\label{Section_Conclusion}

The new version of FeynGrav implements advanced developments in techniques for calculating Feynman rules. Firstly, we developed a BRST formulation of gravity that results in a finite set of rules for the Faddeev-Popov ghost sector. Such a simplification is possible because one can separate the background geometry from the metric perturbations that propagate on it. In turn, the gauge transformations split into two parts. One corresponds to the freedom to choose the reference frame on the background spacetime. The other one corresponds to the true gauge redundancy of the gravitational field description. Secondly, we implemented the gravity formulation in terms of the Cheung-Remmen variables. These variables admit a non-linear relation with the standard metric and bring the Hilbert action to a polynomial form. In turn, one can obtain a finite set of Feynman rules for these variables, which dramatically simplifies all the computations with the standard Feynman diagrammatic.

Further development of \texttt{FeynGrav} clearly appears in the light of the presented results. First and foremost, we shall look for a way to bring an arbitrary gravity theory to a finite set of rules with the use of a generalisation of Cheung-Remmen variables. If such variables are found, the technical side of the computations in quantum gravity will be significantly simplified. Although the discovery of such variables will not resolve the problem of non-renormalizability, they will strongly advance all the computations.

Secondly, we will search for ways to improve \texttt{FeynGrav} performance. In the present version, the most resource-demanding processes of \texttt{FeynGrav} are related to the generation of symmetric expressions for vertices. On the other hand, there is at least one algorithm \cite{Portugal:1998qi} that can significantly reduce the computational time for all the procedures implemented in \texttt{FeynGrav}.

Last but not least, we shall pursue the development and implementation of the other gravity model. The two main candidates are massive gravity and supergravity. The massive gravity is perspective because the corresponding set of interaction rules can help obtain new exact results, which can be used to establish constraints on the graviton mass. Supergravity is considered a perspective because it is a gravitational model with better renormalisation behaviour than general relativity. The opportunity to work with the corresponding set of Feynman rules will help to compare two theories and study their difference directly.

\section*{Acknowledgment}
The work was supported by the Foundation for the Advancement of Theoretical Physics and Mathematics “BASIS”.

\appendix

\section{Nieuwenhuizen operators}\label{Appendix_Nieuwenhuizen}

Nieuwenhuizen operators \cite{VanNieuwenhuizen:1973fi,Hindawi:1995an,Accioly:2000nm} are generalisations of the standard gauge projectors. There are two gauge projectors for gauge models with spin $1$:
\begin{align}
    \theta_{\mu\nu}(p) &\overset{\text{def}}{=} \eta_{\mu\nu} - \cfrac{p_\mu p_\nu}{p^2} \, ,& \omega_{\mu\nu}(p) &\overset{\text{def}}{=} \cfrac{p_\mu p_\nu}{p^2} \, , & \theta_{\mu\nu}(p) + \omega_{\mu\nu}(p) & = \eta_{\mu\nu} .
\end{align}
We use the following definition of the Nieuwenhuizen operators:
\begin{align}
    \begin{split}
        P^1_{\mu\nu\alpha\beta}(p) &= \cfrac{1}{2} \left[ \theta_{\mu\alpha}(p) \omega_{\nu\beta}(p) + \theta_{\mu\beta}(p) \omega_{\nu\alpha}(p) + \theta_{\nu\alpha}(p) \omega_{\mu\beta}(p) + \theta_{\nu\beta}(p) \omega_{\mu\alpha}(p) \right] , \\ 
        P^2_{\mu\nu\alpha\beta}(p) &= \cfrac{1}{2} \left[ \theta_{\mu\alpha}(p) \theta{\nu\beta}(p) + \theta_{\mu\beta}(p) \theta{\nu\alpha}(p)  \right] - \cfrac13 \, \theta_{\mu\nu}(p) \theta_{\alpha\beta}(p) ,\\
        P^0_{\mu\nu\alpha\beta}(p) &= \cfrac13 \, \theta_{\mu\nu}(p) \theta_{\alpha\beta}(p) , \\
        \overline{P}^0_{\mu\nu\alpha\beta}(p) &= \omega_{\mu\nu}(p) \omega_{\alpha\beta}(p) , \\
        \overline{\overline{P}}^0_{\mu\nu\alpha\beta}(p) &=  \theta_{\mu\nu}(p) \omega_{\alpha\beta}(p) + \omega_{\mu\nu}(p) \theta_{\alpha\beta}(p) .
    \end{split}
\end{align}
Four of these operators form a complete set of projectors:
\begin{align}
    P^1_{\mu\nu\alpha\beta}(p) + P^2_{\mu\nu\alpha\beta}(p) + P^0_{\mu\nu\alpha\beta}(p) + \overline{P}^0_{\mu\nu\alpha\beta}(p) = I_{\mu\nu\alpha\beta} \,.
\end{align}
The operator $\overline{\overline{P}}^0$ is introduced purely for the sake of convenience and can be safely omitted. It is orthogonal to $P^1$ and $P^2$ and admits the following relations:
\begin{align}
    \begin{split}
        \overline{\overline{P}}^0_{\mu\nu\rho\sigma} \overline{\overline{P}}^0{}^{\rho\sigma}{}_{\alpha\beta} & = 3 \left( P^0 + \overline{P}^0 \right) , \\
        P^0_{\mu\nu\rho\sigma} \overline{\overline{P}}^0{}^{\rho\sigma}{}_{\alpha\beta} = \overline{\overline{P}}^0_{\mu\nu\rho\sigma} \overline{P}^0{}^{\rho\sigma}{}_{\alpha\beta} & = \theta_{\mu\nu} \, \omega_{\alpha\beta} , \\
        \overline{P}^0_{\mu\nu\rho\sigma} \overline{\overline{P}}^0{}^{\rho\sigma}{}_{\alpha\beta} = \overline{\overline{P}}^0_{\mu\nu\rho\sigma} P^0{}^{\rho\sigma}{}_{\alpha\beta} & = \omega_{\mu\nu} \, \theta^{\alpha\beta} .
    \end{split}
\end{align}
These relations allow one to obtain a formula for the inversion of an operator. If operator $\mathcal{O}$ reads
\begin{align}
    \mathcal{O}_{\mu\nu\alpha\beta} = z_1 P^1_{\mu\nu\alpha\beta} + z_2 P^2_{\mu\nu\alpha\beta} + z_0 P^0_{\mu\nu\alpha\beta} + \overline{z}_0 \overline{P}^0_{\mu\nu\alpha\beta} + \overline{\overline{z}}_0 \overline{\overline{P}}^0_{\mu\nu\alpha\beta}, 
\end{align}
then the inverse operator reads \cite{Accioly:2000nm}:
\begin{align}
    \begin{split}
        (\mathcal{O}^{-1})_{\mu\nu\alpha\beta} =& \cfrac{1}{z_1} P^1_{\mu\nu\alpha\beta} + \cfrac{1}{z_2} P^2_{\mu\nu\alpha\beta} + \cfrac{1}{z_2} \cfrac{ (d-4) ( z_0 \overline{z}_0 - 3 \, \overline{\overline{z}}_0^2) - (d-7) \, z_2 \overline{z}_0   }{ (d-1) ( z_0 \overline{z}_0 - 3 \, \overline{\overline{z}}_0^2) - (d-4) \, z_2 \overline{z}_0 } P^0_{\mu\nu\alpha\beta} \\
        & + \cfrac{ (d-1) z_0 - (d-4) \, z_2  }{ (d-1) ( z_0 \overline{z}_0 - 3 \, \overline{\overline{z}}_0^2) - (d-4) \, z_2 \overline{z}_0 } \overline{P}^0_{\mu\nu\alpha\beta} - \cfrac{ 3 \, \overline{\overline{z}}_0 }{ (d-1) ( z_0 \overline{z}_0 - 3 \, \overline{\overline{z}}_0^2) - (d-4) \, z_2 \overline{z}_0 } \, \overline{\overline{P}}^0_{\mu\nu\alpha\beta}. 
    \end{split}
\end{align}

\section{Command reference}

This appendix provides tables summarising the commands implemented in \texttt{FeynGrav}. The typical workflow is to import the required library of vertices for a chosen model, use those vertices together with propagators and assemble amplitudes or loop integrands. For convenience, the tables below collect the relevant commands, with short descriptions, and indicate the default values of optional arguments and the gauge-fixing parameters that enter the propagators. We also plan to provide an online reference that presents up-to-date documentation for all \texttt{FeynGrav} commands and options.

\begin{table}[!ht]
  \centering
  \setlength{\tabcolsep}{4pt}
  \renewcommand{\arraystretch}{1.1}
  \begin{tabular}{|l|p{8.1cm}|}
    \hline
    \textbf{Command} & \textbf{Description} \\ \hline

    \texttt{importGravitons[n]} & Imports graviton vertices in general relativity. \\ \hline
    \texttt{importScalars[n]} & Imports graviton--scalar vertices for the minimal coupling. \\ \hline
    \texttt{importFermions[n]} & Imports graviton--Dirac fermion vertices for the minimal coupling. \\ \hline
    \texttt{importVectors[n]} & Imports graviton--vector vertices for the minimal coupling. \\ \hline
    \texttt{importSUNYM[n]} & Imports graviton--$\operatorname{SU}(\operatorname{N})$ Yang--Mills vertices. \\ \hline
    \texttt{importAxionVectorVertex[n]} & Imports graviton--axion--vector interaction vertices. \\ \hline

    \texttt{importHorndeskiG2[]} & Imports vertices for the $G_2$ sector of the Horndeski theory. \\ \hline
    \texttt{importHorndeskiG3[]} & Imports vertices for the $G_3$ sector of the Horndeski theory. \\ \hline
    \texttt{importHorndeskiG4[]} & Imports vertices for the $G_4$ sector of the Horndeski theory. \\ \hline
    \texttt{importHorndeskiG5[]} & Imports vertices for the $G_5$ sector of the Horndeski theory. \\ \hline
    \texttt{importScalarGaussBonnet[n]} & Imports scalar--Gauss--Bonnet interaction vertices. \\ \hline
    \texttt{importQuadraticGravity[n]} & Imports graviton vertices in quadratic gravity. \\ \hline
  \end{tabular}

  \renewcommand{\arraystretch}{1}
  \vspace{0.5ex}
  \parbox{0.92\linewidth}{\textit{Note:} Commands import vertices up to $\okappa{n}$. The parameter defaluts to $n=2$. If libraries for the requested order are not available, the command imports up to the highest existing order.}

  \caption{Commands that import libraries for models with minimal coupling to gravity and for selected modified-gravity models.}
  \label{Table_Imports}
\end{table}

\begin{table}[!ht]
  \centering
  \setlength{\tabcolsep}{4pt}
  \renewcommand{\arraystretch}{1.1}
  \begin{tabular}{|l|p{7cm}|}
    \hline
    \textbf{Command} & \textbf{Description} \\ \hline
    \texttt{GravitonScalarVertex} & Returns graviton--scalar vertices from the scalar kinetic term. \\ \hline
    \texttt{GravitonScalarPotentialVertex} & Returns graviton--scalar vertices from the scalar potential term. \\ \hline
    \texttt{GravitonFermionVertex} & Returns graviton--Dirac fermion vertices. \\ \hline
    \texttt{GravitonMassiveVectorVertex} & Returns graviton--massive-vector vertices (Proca field). \\ \hline
    \texttt{GravitonVectorVertex} & Returns graviton--massless-vector vertices. \\ \hline
    \texttt{GravitonVectorGhostVertex} & Returns graviton--Faddeev--Popov-ghost vertices for a massless vector field. \\ \hline
    \texttt{GravitonGluonVertex} & Returns graviton couplings to two, three, and four gluons. \\ \hline
    \texttt{GravitonQuarkGluonVertex} & Returns graviton contributions to the quark--gluon interaction vertex. \\ \hline
    \texttt{GravitonYMGhostVertex} & Returns graviton--Faddeev--Popov-ghost vertices in Yang--Mills theory (gluons). \\ \hline
    \texttt{GravitonGluonGhostVertex} & Returns graviton vertices involving gluons and Yang--Mills ghosts. \\ \hline
    \texttt{GravitonVertex} & Returns graviton self-interaction vertices in general relativity. \\ \hline
    \texttt{GravitonGhostVertex} & Returns graviton--Faddeev--Popov-ghost vertices in general relativity. \\ \hline
    \texttt{GravitonGhostVertexHD} & Returns graviton--Faddeev--Popov-ghost vertices for the higher-derivative gauge fixing term. \\ \hline
    \texttt{GravitonAxionVectorVertex} & Returns graviton vertices for an axion-like scalar coupling to one vector field. \\ \hline
    \texttt{HorndeskiG2} & Returns vertices for the Horndeski $G_2$ sector. \\ \hline
    \texttt{HorndeskiG3} & Returns vertices for the Horndeski $G_3$ sector. \\ \hline
    \texttt{HorndeskiG4} & Returns vertices for the Horndeski $G_4$ sector. \\ \hline
    \texttt{HorndeskiG5} & Returns vertices for the Horndeski $G_5$ sector. \\ \hline
    \texttt{ScalarGaussBonnet} & Returns vertices for scalar--Gauss--Bonnet interactions (scalars and gravitons). \\ \hline
    \texttt{QuadraticGravityVertex} & Returns graviton self-interaction vertices in quadratic gravity. \\ \hline
  \end{tabular}
  \renewcommand{\arraystretch}{1}
  \caption{List of commands describing interaction vertices (excluding Cheung--Remmen variables).}\label{Table_Vertices}
\end{table}

\begin{table}[!ht]
  \centering
  \setlength{\tabcolsep}{4pt}
  \renewcommand{\arraystretch}{1.1}
  \begin{tabular}{|l|p{7cm}|}
    \hline
    \textbf{Command} & \textbf{Description} \\ \hline
    \texttt{GravitonVertexCRhhh} & Returns Cheung--Remmen vertices for the $\goh\goh\goh$ interaction. \\ \hline
    \texttt{GravitonVertexCRBhh} & Returns Cheung--Remmen vertices for the $B\goh\goh$ interaction. \\ \hline
    \texttt{GravitonVertexCRBBh} & Returns Cheung--Remmen vertices for the $BB\goh$ interaction. \\ \hline
  \end{tabular}
  \renewcommand{\arraystretch}{1}
  \caption{List of commands describing interaction vertices for Cheung--Remmen variables.}\label{Table_Vertices_CR}
\end{table}

\begin{table}[!ht]
  \centering
  \setlength{\tabcolsep}{4pt}
  \renewcommand{\arraystretch}{1.1}
  \begin{tabular}{|l|p{7cm}|}
    \hline
    \textbf{Command} & \textbf{Description} \\ \hline
    \texttt{ScalarPropagator} & Scalar propagator. \\ \hline
    \texttt{ProcaPropagator} & Propagator of a massive vector field. \\ \hline
    \texttt{GravitonPropagator} & Graviton propagator. \\ \hline
    \texttt{GravitonPropagatorMassive} & Massive graviton propagator. \\ \hline
    \texttt{QuadraticGravityPropagator} & Graviton propagator in quadratic gravity. \\ \hline
    \texttt{GhostVectorPropagator} & Faddeev--Popov vector ghost propagator. \\ \hline
    \texttt{QuadraticGravityPropagatorHD} & Quadratic-gravity graviton propagator with higher-derivative gauge fixing. \\ \hline
    \texttt{GhostVectorPropagatorHD} & Vector ghost propagator for the higher-derivative gauge fixing term. \\ \hline
    \texttt{GravitonPropagatorCR} & Cheung--Remmen propagator for the perturbative field $\goh^{\mu\nu}$. \\ \hline
    \texttt{GravitonPropagatorAuxiliaryCR} & Cheung--Remmen propagator for the auxiliary field $B^\lambda_{\alpha\beta}$ (algebraic; momentum-free). \\ \hline
  \end{tabular}

  \renewcommand{\arraystretch}{1}
  \vspace{0.5ex}
  \parbox{0.92\linewidth}{\textit{Note:}
  The gauge-fixing parameter \texttt{GaugeFixingEpsilon} enters only propagators in the GR and quadratic-gravity ghost sector.
  For the higher-derivative gauge fixing term, propagators depend on \texttt{GaugeFixingEpsilonHD}, \texttt{GaugeFixingEpsilonHD0}, and \texttt{GaugeFixingEpsilonHD1}.
  For Cheung--Remmen variables, the gauge-fixing parameter \texttt{GaugeFixingEpsilonCR} affects only \texttt{GravitonPropagatorCR}.}

  \caption{List of commands for propagators.}\label{Table_Propagators}
\end{table}

\clearpage
\bibliographystyle{elsarticle-num}
\bibliography{FG4.bib}

\end{document}